\documentclass[letter,reqno]{a_preamble}
\title{
Periodicity in Ergodic Quantum Processes
}

\date{}

\usepackage{authblk}
\usepackage[backend=biber]{biblatex}
\author[1, 2]{Owen Ekblad\orcidlink{0009-0006-0834-0327}\thanks{ekbladow@msu.edu}}
\author[1]{Jeffrey Schenker\orcidlink{0000-0002-1171-7977}\thanks{schenke6@msu.edu}}
\affil[1]{Department of Mathematics, Michigan State University}
\affil[2]{Department of Mathematics and Center for Quantum Mathematics and Physics, University of California, Davis}

\begin{document}

\pagenumbering{arabic}
\lhead{\thepage}
\maketitle

\begin{abstract}
We study the periodic properties of sequences of quantum channels sampled from an ergodic stochastic process satisfying a natural irreducibility condition. 
We relate these periodic properties to certain global spectral data defined by the sequence of quantum channels, proving a general Perron-Frobenius-type theorem that in particular constitutes a full generalization of the work of Evans and Høegh-Krohn to the disordered setting.  
\end{abstract}
\section{Introduction}\label{Sec:Introduction}
Originally motivated by problems in generalized continued fractions \cite{Perron1907ZurMatrices, Frobenius1912UberElementen, Hawkins2008ContinuedTheorem}, Perron-Frobenius (\textit{PF}) theory is now widespread across many branches of both pure and applied mathematics.
At the most abstract level, PF theory relates spectral properties of positive operators to dynamical interpretations they admit.
A well-known application of PF theory is to the study of Markov chains on finite state spaces, where the basic objects are stochastic matrices, i.e., matrices whose entries are all positive and whose rows sum to 1 \cite[Ch. I]{Schaefer1974BanachOperators}.
By understanding completely positive and trace preserving maps on matrix algebras (henceforth referred to as \textit{quantum channels}) as a noncommutative generalization of stochastic matrices, PF-theoretic methods can be applied to quantum channels, which has been a fruitful perspective taken by many authors since at least the 1970s \cite{Evans1978SpectralC-Algebras, Albeverio1978FrobeniusAlgebras, Enomoto1979AC-algebras, Groh1981TheC-Algebras, Groh1983OnC-Algebras, Groh1984UniformW-Algebras, Groh1984UniformlyC-algebras, Farenick1996IrreducibleAlgebras,  Schrader2001Perron-FrobeniusIdeals, Fagnola2009IrreducibleMaps}.
Owed to the central position held by quantum channels in quantum theory and applications \cite{Davies1976QuantumSystems, Nielsen2012QuantumInformation, Watrous2018TheInformation}, the PF theory of quantum channels has found a natural application to the study of open quantum systems \cite{Kummerer2003AnProcesses, Sanz2010AInequality, Wolf2012QuantumTour, Burgarth2013ErgodicDimensions, Carbone2016OpenProperties, Carbone2020OnChannel}, where spectral and dynamical assumptions about quantum channels serve fundamental purposes, and quantum channels are also a central object in the theory of matrix product states \cite{Fannes1992FinitelyChains, Perez-Garcia2007MatrixRepresentations}.
Most PF theory for quantum channels in the literature is concerned with single, fixed quantum channels, which is insufficient for understanding many physically-interesting situations \cite{Bruneau2014RepeatedSystems, Ciccarello2022QuantumInteractions, Movassagh2021TheoryProcesses, Movassagh2022AnStates}. 
In this work, our focus is on developing the PF theory of \textit{sequences} $\eqp = \seq{\phi_n}_{n\in\mbZ}$ of quantum channels sampled from an ergodic stochastic process, which are called ergodic quantum process. 
Ergodic quantum processes are useful for describing matrix products states on disordered quantum spin chains \cite{Movassagh2022AnStates, Ekblad2026ParentStates}, random repeated quantum interactions \cite{Bruneau2008RandomSystems, Nechita2012RandomStates, Ciccarello2022QuantumInteractions, Bougron2022MarkovianSystems, Ekblad2025AEntanglement}, and quantum trajectories under generalized disordered measurements \cite{Ekblad2025AsymptoticMeasurements, Ekblad2026ErgodicMeasurements}.
The basic method of PF theory is the algebraic reduction of dynamics to irreducible components, where dynamical properties can then be understood via equilibrium states.
In the companion work \cite{Ekblad2024ReducibilityProcesses}, the authors of the present work laid the foundations of a PF theory for ergodic quantum processes, defining and characterizing a notion of irreducibility.
Here, we continue in along these lines, giving an in-depth analysis of irreducible ergodic quantum processes and proving Perron-Frobenius-type theorems, which constitute a full generalization of the influential work of Evans and Høegh-Krohn \cite{Evans1978SpectralC-Algebras} on irreducible quantum channels. 
We characterize dynamical properties of irreducible ergodic quantum processes in terms of associated spectral data, which we show is succinctly phrased in terms of a certain finite abelian group $\GammaGroup$ (Theorem \ref{Thm:Finite_group}).
For the nonrandom situation, $\GammaGroup$ is always a cyclic group, equal to the set of peripheral eigenvalues of the irreducible channel, but here we demonstrate how the underlying disordered allows $\GammaGroup$ to be \textit{any} finite abelian group (Theorem \ref{Thm:Any_group}).
We then describe how $\GammaGroup$ encodes the periodicity of the ergodic quantum process (Theorems \ref{Thm:Partition}, \ref{Thm:Periodicity}, and \ref{Thm:PeriodicChar}).
Let us now describe our results in more detail.
\subsection{Perron-Frobenius theory}
We begin by describing the (non-disordered) PF theory for quantum channels acting on finite dimensional $C^*$-algebras as was established by Evans and Høegh-Krohn in the influential paper \cite{Evans1978SpectralC-Algebras}.
For simplicity, we describe this theory in terms of $\matrices$, the set of $d\times d$ matrices with complex-valued entries.
In this setting, the basic notion of reducibility is phrased in terms of projections, where we recall that a matrix $p\in\matrices$ is called a projection if $p\neq 0$ and $p = p^2 = p^*$, where $p^*$ denotes the conjugate transpose. 
Then a projection $p\in\matrices$ is called a reducing projection for a quantum channel $\phi:\matrices\to\matrices$ if
\begin{equation}
    \phi\!\seq{p\matrices p}\subseteq p\matrices p,
\end{equation}
and we say that $p$ is a minimal reducing projection for $\phi$ if $p$ is minimal among the set of reducing projections with respect to the usual ordering $\leq$ on positive semidefinite matrices. 
If the only minimal reducing projection is $p = \mbI$, the identity matrix, then $\phi$ is said to be irreducible.
In \cite{Evans1978SpectralC-Algebras}, it was shown that irreducibility of $\phi$ is equivalent to the existence and uniqueness of a positive definite density matrix\footnote{Recall a matrix is called a density matrix if it is positive semidefinite and has trace 1.} $\varrho\in\matrices$ with $\phi(\varrho) = \varrho$ \cite{Evans1978SpectralC-Algebras}, a result analogous to the characterization of irreducible Markov chains in terms of unique invariant measures. 
We call this unique $\varrho$ the \textit{Perron-Frobenius eigenmatrix} of $\phi$.  
The work of Evans and Høegh-Krohn most relevant to our goals is the following amalgamation of results from \cite{Evans1978SpectralC-Algebras}.\footnote{These results were proved for a slightly more general class of positive maps on $\matrices$ (the so-called Schwarz maps), but we state the results here as they apply to quantum channels for simplicity.} 
Below, we call a finite set $\set{p_k}_{k\in I}\subset\matrices$ of projections a partition of unity if $\sum_{k\in I}p_k = \mbI$ and $p_kp_j = 0$ for all $k\neq j$.
\begin{thmx}[Evans and Høegh-Krohn, 1978]\label{Thm:EHK}
    Let $\phi:\matrices\to\matrices$ be an irreducible quantum channel with Perron-Frobenius eigenmatrix $\varrho$, let $\Gamma_\phi$ be the set of eigenvalues $\alpha$ of $\phi$ with $|\alpha| = 1$, and let $m = |\Gamma_\phi|$. 
    Then $\Gamma_\phi$ is the group of $m$th roots of unity, every $\alpha\in\Gamma_\phi$ is a simple eigenvalue, and $m \leq d$. 
    Further, there is a partition of unity $\set{p_k}_{k\in\mbZ/m\mbZ}$ such that
    \begin{equation}
        \phi(p_k\matrices p_k)\subseteq p_{k+1}\matrices p_{k+1}
    \end{equation}
    for all $k\in\mbZ/m\mbZ$ and  
    \begin{equation}
        \varrho 
        =
        \cfrac{1}{m}
        \sum_{k\in\mbZ/m\mbZ}
        \cfrac{p_k\varrho p_k}{\tr{p_k\varrho p_k}}.
    \end{equation}
    Moreover, $p_k$ is a minimal reducing projection for $\phi^m$, and, in particular, $|\Gamma_\phi| = 1$ if and only if $\phi^n$ is irreducible for all $n\in\mbN$.
\end{thmx}
Our goal is to extend this theorem to the case of ergodic quantum processes. 
Towards this end, let us enumerate three aspects of Theorem \ref{Thm:EHK} that we will focus on individually in the disordered setting.
\begin{description}
    \hypertarget{Group}{}
    \item[\Group] Peripheral eigenspectrum consists of simple eigenvalues which together have group-theoretic structure. 

    \hypertarget{Proj}{}
    \item[\Proj] Group-theoretic structure of the peripheral eigenspectrum describes periodicity in the dynamics by means a partition of unity.

    \hypertarget{Period}{}
    \item[\Period] The projections making up this partition of unity reduce an iterated version of the process, which fully characterizes triviality of eigenspectrum.
\end{description}
As we shall see, the extension of each of these aspects of Theorem \ref{Thm:EHK} to the setting of ergodic quantum processes requires new techniques. 
Let us begin by precisely defining what we mean by a quantum process. 
\begin{definition}[Disordered quantum processes]\label{Def:EQP_defn_introduction}
    Let $\seq{\Omega, \mcF, \mu}$ be a probability space, let $\theta:\Omega\to\Omega$ be a measurable map, and let $\phi:\Omega\to\set{\text{quantum channels}}$ be a random quantum channel. 
    Then the \textit{disordered quantum process} defined by $\seq{\theta, \phi}$ is the sequence $\eqp = \seq{\phi_n}_{n\in\mbN}$ of random quantum channels $\phi_n$ defined by $\omega\mapsto \phi_{\theta^n(\omega)} =: \phi_{n; \omega}.$
    In the case that $\theta$ is an invertible ergodic $\mu$-preserving transformation, we call $\eqp$ an ergodic quantum process. 
\end{definition}
Given a random projection $p:\Omega\to\matrices$, we say that $p$ reduces the disordered quantum process defined by $\seq{\theta, \phi}$ if 
\begin{equation}
    \phi_{\theta(\omega)}\!\seq{p_\omega \matrices p_\omega}
    \subseteq 
    p_{\theta(\omega)}\matrices p_{\theta(\omega)}
\end{equation}
holds for almost every (\textit{a.e.}) $\omega\in\Omega$, and we call $\eqp$ irreducible if the only $p$ reducing $\eqp$ is $p = \mbI$. 
Now, let $\eqp$ be a fixed \textit{ergodic} quantum process defined by $\seq{\theta, \phi}$. 
In \cite{Ekblad2024ReducibilityProcesses}, the following characterization of irreducibility was given. 
\begin{thmx}[\texorpdfstring{\cite{Ekblad2024ReducibilityProcesses}}{l}]\label{Thmx:irreducibility}
    Let $\eqp$ be an ergodic quantum process. 
    Then $\eqp$ is irreducible if and only if there is a unique random density matrix $\varrho:\Omega\to\matrices$ such that $\varrho_\omega >0$ and
    \begin{equation}\label{Eqn:Intro:Steady_state}
        \phi_{\theta(\omega)}\!\seq{\varrho_\omega}
        =
        \varrho_{\theta(\omega)}
    \end{equation}
    holds for a.e. $\omega\in\Omega$.
\end{thmx}
We call the random density matrix $\varrho:\Omega\to\matrices$ in the above theorem the \textit{unique steady state} of $\eqp$.
Interpreting (\ref{Eqn:Intro:Steady_state}) as an eigenequation for $\eqp$ with eigenvalue $\alpha = 1$, the uniqueness of $\varrho$ in the above theorem may be understood as saying that $\alpha = 1$ is a simple eigenvalue of $\eqp$. 
In general, we let $\SpecEQP$ denote the set of all $\alpha\in\mbC$ for which there exists a random matrix $x:\Omega\to\matrices$ such that 
\begin{equation}\label{Eqn:Intro:Eigenequation}
        \phi_{\theta(\omega)}\!\seq{x_\omega}
        =
        \alpha x_{\theta(\omega)}
\end{equation}
holds for a.e. $\omega\in\Omega$, and we let $\PerSpecEQP = \set{\alpha\in\SpecEQP\,\,:\,\,|\alpha| = 1}$. 
We call an element of $\SpecEQP$ an eigenvalue of $\eqp$, we call an element of $\PerSpecEQP$ a peripheral eigenvalue of $\eqp$, and we say $\alpha\in\SpecEQP$ is \textit{simple} if the vector space of random matrices satisfying (\ref{Eqn:Intro:Eigenequation}) is 1-dimensional.
With this in hand, we can begin to see how \Group\assumptionspace applies to $\eqp$. 
Our first result towards this end is that $\PerSpecEQP$ indeed forms a group under multiplication (Corollary \ref{Cor:Lambda_is_a_group}). 
However, this result alone is not enough to satisfactorily address \Group.
Indeed, to see this, the reader may consider two simple but markedly different situations that both fall within our theoretical framework: 
the case that $\theta$ is weakly mixing\footnote{A probability preserving transformation $T:\Omega\to\Omega$ is called weakly mixing if for all measurable sets $A, B\in\mcF$, $\lim_{N}N^{-1}\sum_{n=1}^N\abs{\prob{T^{-n}(A)\cap B} - \prob{A}\prob{B}} = 0$.}, and the case that $\theta$ is the irrational rotation of the circle\footnote{In this case, $\Omega$ is $\mbT$, the unit circle in $\mbC$, and $\theta$ is the map $\theta(\omega) = \gamma\omega$, where $\gamma = e^{2\pi i t}$ for some irrational $t\in\mbR$.}.
One of the invariants that discerns between these two regimes is the set $\Lambda_\theta$ of eigenvalues of $\theta$.\footnote{We call $\beta\in\mbC$ an eigenvalue of $\theta$ if there is a nonzero measurable function $f:\Omega\to\mbC$ such that $f\circ\theta = \beta f$ a.s.}
In the weakly mixing case, $\Lambda_\theta = \set{1}$, but in the quasiperiodic case, $\Lambda_\theta = \set{\gamma^n}_{n\in\mbZ}$ where $\gamma = e^{2\pi i t}$ for some irrational number $t\in\mbR$. 
In both cases, however, we show that $\Lambda_\theta\subseteq \PerSpecEQP$ (Lemma \ref{Lem:Eigenspectrum_koopman_in_opL}), and we show later that $\Lambda_\theta$ does not give rise to any partitions of unity as in \Proj. 
Therefore, it is necessary to understand precisely what role $\Lambda_\theta$ plays as a subset of $\PerSpecEQP$.
To do this, we note the standard fact that $\Lambda_\theta$ is itself a group  \cite[Theorem 3.1]{Walters1982AnTheory}, so, to understand how the group-theoretic structure of $\PerSpecEQP$ influences periodicity of $\eqp$, the natural object to consider is the quotient group
\begin{equation}
    \GammaGroup 
        :=
    \PerSpecEQP/\Lambda_\theta.
\end{equation}
Our first technical result gives the structure of $\GammaGroup$.
For $\alpha\in\PerSpecEQP$, we let $N_\alpha = \inf\set{n\in\mbN\,\,:\,\, \alpha^n\in\Lambda_\theta}$, where $\inf\emptyset = \infty$ by convention.
\begin{restatable}{thm}{FiniteGroup}\label{Thm:Finite_group}
    Let $\eqp$ be an irreducible ergodic quantum process. 
    Then $\GammaGroup$ is a finite abelian group with $|\GammaGroup|\leq d^2$, and every $\alpha\in \PerSpecEQP$ is a simple eigenvalue with $N_\alpha \leq d$. 
\end{restatable}
A natural follow-up question to this theorem pertains to the structure of $\GammaGroup$ itself. 
In Theorem \ref{Thm:EHK}, \Group\assumptionspace manifests in a stronger form than what we have written here: namely, they conclude that the peripheral eigenspectrum has the structure of a cyclic group. 
Here, we find a departure from this simple situation, discovering that any finite abelian group may be realized as $\GammaGroup$ for some suitable ergodic quantum process. 
\begin{restatable}{thm}{AnyGroup}\label{Thm:Any_group}
    For any finite abelian group $G$, there is an irreducible ergodic quantum process $\Phi$ with $\GammaGroup\cong G$. 
\end{restatable}
To prove this theorem, we explicitly construct an ergodic quantum process with the properties of the theorem, which in particular shows that the bounds in Theorem \ref{Thm:Finite_group} are sharp. 
Under certain assumptions about $\theta$ (which hold, for example, when $\theta$ is weakly mixing), we can prove the stronger fact that $|\GammaGroup|\leq d$ and that $\GammaGroup$ is a cyclic group (Corollary \ref{Cor:theta_torsion_group_order_leq_d}). 
It is therefore worth noting that our construction of $\Phi$ as in Theorem \ref{Thm:Any_group} relies on $\theta$ being non-chaotic---specifically, the measure-theoretic entropy of $\theta$ in our construction is zero---so the non-cyclicity of $\GammaGroup$ relies on structural properties of $\theta$. 
Theorem \ref{Thm:Finite_group} now gives us a powerful tool for showing how \Proj\assumptionspace arises in the disordered setting. 
First, we require some notation and terminology. 
We say that a finite set $\set{p_k}_{k\in I}$ of random projections is a random partition of unity if $\set{p_{k; \omega}}_{k\in I}$ is a partition of unity for a.e. $\omega\in\Omega$, and, if $\iota$ is a random element of $I$, we let $p_{\iota}:\Omega\to\matrices$ be the random projection defined by $p_{\iota; \omega} 1_{\set{\iota = k}}(\omega) := p_{k; \omega}$.
Then we have the following.
\begin{restatable}{thm}{Partition}\label{Thm:Partition}
    Let $\eqp$ be an irreducible ergodic quantum process with unique steady state $\varrho:\Omega\to\matrices$, and fix $\alpha\in\PerSpecEQP$.
    There is a random partition of unity $\set{p_k}_{k\in\mbZ/N_\alpha\mbZ}$ and measurable $\varsigma:\Omega\to\mbZ/N_\alpha\mbZ$ such that 
    \begin{equation}\label{Eqn:Partition_1}
        \phi_{\theta(\omega)}\!\seq{p_{k; \omega}\matrices p_{k; \omega}}
        \subseteq
        p_{k - \varsigma; \theta(\omega)}\matrices p_{k - \varsigma; \theta(\omega)}
    \end{equation}
    for all $k\in\mbZ/N_\alpha\mbZ$ and 
    \begin{equation}\label{Eqn:Partition_2}
        \varrho_\omega 
        =
        \cfrac{1}{N_\alpha}
        \sum_{k\in\mbZ/N_\alpha\mbZ}
        \cfrac{p_{k; \omega}\varrho_\omega p_{k; \omega}}{\tr{p_{k; \omega}\varrho_\omega p_{k; \omega}}}
    \end{equation}
    holds for a.e. $\omega\in\Omega$. 
\end{restatable}
We write $\mcP_\alpha$ and $\varsigma_\alpha$ to denote $\set{p_k}_{k\in\mbZ/N_\alpha\mbZ}$ and $\varsigma$ in the above theorem, respectively. 
This theorem says that in order to generalize \Proj\assumptionspace to the disordered setting, we must accommodate the possibility of disorder-dependent periodicity in the random partition of unity, which, as Example \ref{Example:quasiperiodic} shows, is described by the in general nondeterministic $\varsigma$.
Nevertheless, this disorder-dependent periodicity has some regularity, which leads us to the disordered analog of \Period. 
Given measurable $\tau:\Omega\to\mbN$, we define  $\theta^\tau:\Omega\to\Omega$ by $ \theta^\tau(\omega) = \theta^n(\omega)$ on the event $\set{\tau = n}$---which is measurable by the measurability of $\tau$---and we let $\phi^{(\tau)}$ be the random quantum channel
\begin{equation}
    \phi^{(\tau)}_\omega
        :=
    \phi_{\theta^\tau(\omega)}\circ\cdots\circ\phi_{\theta(\omega)}.
\end{equation}
%
% If $\set{\tau = n}$ is measurable with respect to the $\sigma$-algebra $\sigma(\phi_1, \dots, \phi_n)$ for all $n$, we call $\tau$ a $\eqp$-stopping time. 
%
% We let $\taueqp{\tau}$ be the disordered quantum process defined by $\seq{\theta^\tau, \phi^{(\tau)}}$. 
%
For $n\in\mbN$, we define $\tau_n$ iteratively by $\tau_1 := \tau$ and $\tau_{n+1} = \tau\circ\theta^{\tau_{n}}$ for $n\geq 1$, and we say that $\tau$ has density $\epsilon>0$ if 
\begin{equation}\label{Eqn:Intro:Density_of_tau}
    \liminf_{N\to\infty}
    \cfrac{
    \#\seq{\set{\tau_n(\omega)}_{n = 1}^\infty\cap\set{1, \dots, N}}
    }{
    N
    }
    =
    \limsup_{N\to\infty}
    \cfrac{
    \#\seq{\set{\tau_n(\omega)}_{n = 1}^\infty\cap\set{1, \dots, N}}
    }{
    N
    }
    =
    \epsilon
\end{equation}
holds for $\mu$-almost every $\omega\in\Omega$. 
We say $\tau$ is a $\eqp$-stopping time if $\set{\tau = n}$ is measurable with respect to $\sigma(\phi_1, \dots, \phi_n)$ for all $n$. 
\begin{restatable}{thm}{Periodicity}\label{Thm:Periodicity}
    Let $\eqp$ be an irreducible ergodic quantum process, fix $\alpha\in\PerSpecEQP$, and let $p\in\mcP_\alpha$. 
    Then there is a $\eqp$-stopping time $\tau$ of density $N_\alpha^{-1}$ such that
    \begin{equation}\label{Eqn:Periodicity}
        \phi^{(\tau)}_\omega\seq{p_\omega\matrices p_\omega}\subseteq p_{\theta^\tau(\omega)}\matrices p_{\theta^\tau(\omega)}
    \end{equation}
    holds for a.e. $\omega\in\Omega$. 
    In particular, if $|\GammaGroup|>1$, then the disordered quantum process defined by $\seq{\theta^\tau, \phi^{(\tau)}}$ is reducible.
\end{restatable}
This partially answers \Period, but leaves open the converse classification of when processes like $\seq{\theta^\tau, \phi^{(\tau)}}$ are reducible. 
To deal with this, we adjust our perspective of the problem via \textit{extension of the probability space}: 
in Theorem \ref{Thm:EHK}, it holds that $\phi^n$ is irreducible for all $n$ if and only if $|\Gamma_\phi| = 1$. 
From \cite{Evans1978SpectralC-Algebras}, an equivalent way to phrase this is that for all $n\in\mbN$, if we define 
\begin{equation*}
    \mbZ/n\mbZ\to\channels,\quad k\mapsto \phi \text{ for all $k\in\mbZ/n\mbZ$,}
\end{equation*}
the corresponding ergodic quantum process defined by making $\mbZ/n\mbZ$ into a probability space with the uniform measure $\operatorname{Haar}$ and the map $\theta(k) = k + 1$ is irreducible for all $n$. 
This is the formulation of \Period\, that generalizes most naturally in the present setting. 
Given $n\in\mbN$ and $\zeta:\Omega\to\mbZ/n\mbZ$, we say the skew product $\theta_\zeta:\Omega\times \mbZ/n\mbZ\ni (\omega, k)\mapsto (\theta(\omega), k - \zeta(\omega))$ is a spectrally maximal ergodic extension of $\theta$ by $n$ if $\theta_\zeta$ is ergodic and $|\Lambda_{\theta_\zeta}/\Lambda_\theta| = n$. 
\begin{restatable}{thm}{PeriodicChar}\label{Thm:PeriodicChar}
    Let $\eqp$ be an irreducible ergodic quantum process defined by $\seq{\theta, \phi}$.
    Then $|\GammaGroup| = 1$ if and only if for all $n\in\mbN$ and spectrally maximal ergodic extensions $\theta_\zeta:\Omega\times\mbZ/n\mbZ\to\Omega\times\mbZ/n\mbZ$ of $\theta$ by $n$, the ergodic quantum process $\seq{\theta_\zeta, \phi_\zeta}$ defined by 
    \begin{equation*}
        \phi_{\zeta; (\omega, k)}:= \phi_\omega
    \end{equation*}
    is irreducible. 
    In general, letting $\set{p_k}_{k\in \mbZ/N_\alpha\mbZ}$ and $\varsigma:\Omega\to\mbZ/N_\alpha\mbZ$ be as in Theorem \ref{Thm:Partition}, $\theta_\varsigma$ is a spectrally maximal ergodic extension of $\theta$ by $N_\alpha$, and $(\omega, k)\mapsto p_{k; \omega}$ defines a minimal reducing projection for the ergodic quantum process $\seq{\theta_\varsigma, \phi_\varsigma}$. 
\end{restatable}
%
%
% The nontrivial content of the theorem is proof of the fact that $|\GammaGroup|>1$ implies the existence of such a spectrally maximal ergodic extension of $\theta$ by $n = |\GammaGroup|$ in the case that $|\Lambda_\theta| = \infty$. 
%
% Owed to the fact that $\tau$ is nondeterminstic in general (Example \ref{Example:quasiperiodic}), proving a classification of $|\GammaGroup| = 1$ in terms of reducibility of processes like $\seq{\theta^\tau, \phi^{(\tau)}}$ would require a better understanding of the class of $\eqp$-stopping times $\tau$ arising in the above theorem. 
% %
% We discuss this in more depth in Section \ref{Sec:Final remarks}.
% %
% However, in the setting that $\theta$ is weakly mixing, we have the following refinement of the above theory.
% %
% Recall that $\theta$ is weakly mixing if and only if $\theta^n$ is weakly mixing for all $n\in\mbN$. 
% %
% %
% \begin{restatable}{thm}{WeaklyMixing}\label{Thm:WeaklyMixing}
%     Let $\eqp$ be an irreducible ergodic quantum process defined by $\seq{\theta, \phi}$, assume that $\theta$ is weakly mixing, let $\alpha\in\PerSpecEQP$, and fix $p\in\mcP_\alpha$.
%     %
%     Then $p$ is a minimal reducing projection for the ergodic quantum process defined by $\seq{\theta^{N_\alpha}, \phi^{(N_\alpha)}}$, and, in particular, $|\GammaGroup| = 1$ if and only if $\seq{\theta^n, \phi^{(n)}}$ is irreducible for all $n\in\mbN$. 
% \end{restatable}
%
%
%
\subsection{Relation to other works}
The study of ergodic quantum processes as described in this paper was begun by the pair of papers \cite{Movassagh2021TheoryProcesses, Movassagh2022AnStates}, where the authors were originally motivated by problems in disordered quantum spin chains. 
In this work, Movassagh and the second author of the present work studied $\eqp$ under the assumption of eventual strict positivity, i.e., that for almost every $\omega\in\Omega$, there exists $n\in\mbN$ such that $\phi^{(n)}_\omega(\rho)$ is invertible for all density matrices $\rho\in\matrices$. 
They showed that this is equivalent to a sort of \textit{strong irreducibility}, and we discuss the relation of this concept with the present work in Section \ref{Sec:Final remarks}.
This pair of papers stimulated work on ergodic quantum processes: see \cite{Nelson2024ErgodicAlgebras, Ekblad2025AEntanglement, Souissi2025ErgodicProcesses} for a selection of recent works on ergodic quantum processes.
In general, sequences of random quantum channels was considered earlier by other authors in the repeated interactions literature: see \cite{Bruneau2010InfiniteDynamics, Nechita2012RandomStates, Bruneau2014RepeatedSystems, Bougron2022MarkovianSystems} for a non-exhaustive list of such works. 
As alluded to above, Perron-Frobenius theory has a vast and rich history since its inception over a century ago \cite{Perron1907ZurMatrices, Frobenius1912UberElementen, Hawkins2008ContinuedTheorem}, and we shall not attempt to give a full collection of works that could be interpreted as Perron-Frobenius theory. 
Instead, let us mention a selection of works in PF theory of quantum channels most relevant to the current investigation.
Most directly related to our work is the influential paper of Evans and Høegh-Krohn \cite{Evans1978SpectralC-Algebras} mentioned above, but this paper exists within a wider literature of PF-like results proved for operator algebras.
For a selection of these results, the reader may consult the list \cite{Stormer1963PositiveAlgebras, Albeverio1978FrobeniusAlgebras,Enomoto1979AC-algebras, Groh1981TheC-Algebras, Groh1983OnC-Algebras, Groh1984UniformW-Algebras, Groh1984UniformlyC-algebras, Farenick1996IrreducibleAlgebras, Schrader2001Perron-FrobeniusIdeals, Fagnola2009IrreducibleMaps}.
In the mathematical physics literature, many PF-like results have also been proved.
See \cite{Wolf2012QuantumTour, Burgarth2013ErgodicDimensions, Carbone2016OpenProperties, Carbone2020OnChannel, Carbone2021AbsorptionChannels} for a selection.
\subsection{Organization}
We begin section \ref{Sec:Previous} by giving all the rigorous mathematical notation, terminology, and relevant background results required to give the basic lemmas we use in our proofs. 
Then, we proceed to reformulate the study of $\eqp$ in terms of a certain pair of operators $\opL$ and $\opLdag$, where we show that $\PerSpecEQP$ is precisely equal to the eigenspectrum of $\opL$ and $\opLdag$. 
We then proceed to prove some basic facts about the eigenspectrum of $\opL$ and $\opLdag$ that form the technical backbone of our more involved results. 
Section \ref{Sec:Spectral theory} is devoted to developing our theory in depth, and we give all the technical details of the proofs of our theorems. 
In Section \ref{Sec:Final remarks}, we describe some open problems and future directions, in addition to giving some preliminary results and remarks directed at addressing these problems. 
There are two appendices: one devoted to the disambiguation of the word ``irreducible" in the study of quantum channels, and one for giving proof of the refinement of Theorem \ref{Thm:Finite_group} in the setting that $\Lambda_\theta$ has a particular group-theoretic structure.

%%%

\section{Preliminaries}\label{Sec:Previous}
We shall encounter various Banach, Hilbert, and vector spaces in this work, so let us begin by setting notation and terminology for these objects in full generality.
All vector spaces are assumed to be over the complex numbers $\mbC$. 
Given vector spaces $\scrV$ and $\scrW$, we write $\linears{\scrV, \scrW}$ to denote the set of linear maps $A:\scrV\to\scrW$, and let $\linears{\scrV} = \linears{\scrV, \scrV}$.
For a Banach space $\scrX$ with norm $\banachnorm{\cdot}{\scrX}$ and for $A\in\linears{\scrX}$, we let $\banachnorm{A}{\scrX}$ denote the quantity
\begin{equation}
    \banachnorm{A}{\scrX}
        :=
    \sup_{x\in\scrX\setminus\set{0}}
    \frac{\|Ax\|_\scrX}{\|x\|_\scrX}.
\end{equation}
We let $\bops{\scrX}$ denote the set of all $A\in\linears{\scrX}$ such that $\banachnorm{A}{\scrX}<\infty$.
Let $\dual{\scrX}$ denote the dual of $\scrX$, i.e., $\dual{\scrX}$ is the set of bounded linear functionals $\varphi:\scrX\to\mbC$. 
For $A\in\bops{\scrX}$, let $\dual{A}$ denote the element of $\bops{\dual{\scrX}}$ defined by $\varphi\mapsto \varphi\circ A$. 
We call a subset $\scrP\subseteq\scrX$ a positive cone if $\scrP\cap \seq{-\scrP 
} = \set{0}$, $\scrP + \scrP\subseteq\scrP$, and $\alpha p\in\scrP$ for all $\alpha\in[0, \infty)$. 
For a positive cone $\scrP\subset\scrX$, we write $x\geq_{\scrP} y$ to denote that $x-y\in\scrP$, and, if $\scrP$ is clear from context, we say that $x$ is positive to mean $x\geq_{\scrP} 0$. 
We say that $A\in\linears{\scrX}$ is $\scrP$-positive if $A\seq{\scrP}\subseteq\scrP$.
Let $\identityop{\scrX}\in\bops{\scrX}$ denote the identity operator on $\scrX$.
Given $A\in\bops{\scrX}$, we let $\spectrum{A}$ denote the spectrum of $A$, which is the set 
\begin{equation}
    \spectrum{A}
        :=
    \set{
    \lambda\in\mbC\,\,:\,\,
    A - \lambda \identityop{\scrX}\text{ is not invertible as an element of }\bops{\scrX}
    }.
\end{equation}
We let $\specrad{A}$ denote the spectral radius of $A$, which is the quantity $\sup_{\lambda\in\sigma(A)}|\lambda|$. 
We let $\perspec{A}$ denote the set of $\lambda\in\sigma(A)$ with $|\lambda| = \specrad{A}$, we let $\eigspec{A}$ denote the set of eigenvalues of $A$, and we let $\apspec{A}$ denote the approximate eigenvalues of $A$. 
Let $\scrH$ be a Hilbert space with inner product $\inner{\cdot}{\cdot}_\scrH$ and $\banachnorm{\cdot}{\scrH}$ the corresponding norm. 
%
% %
Given $A\in\bops{\scrH}$, there is a unique operator $A^*\in\bops{\scrH}$ such that for all $f, g\in\scrH$, we have that 
\begin{equation}
    \inner{Af}{g}_\scrH
        =
    \inner{f}{A^*g}_\scrH.
\end{equation}
We call $A^*$ the adjoint of $A$ with respect to $\inner{\cdot}{\cdot}_\scrH$.
The topology generated by the family of seminorms $\set{p_f\,\,:\,\,f\in\scrH}$ where $p_f(A) = \banachnorm{Af}{\scrH}$ is called the strong operator topology (SOT) and the topology generated by the family of seminorms $\set{p_{f, g}\,\,:\,\,f,g\in\scrH}$ where $p_{f, g}(A) = |\inner{f}{Ag}_{\scrH}|$ is called the weak operator topology (WOT).
For facts about functional analysis, one may consult \cite{Conway2007AAnalysis}.
Let $\scrA$ be a $\mbC$-algebra with unit $1_\scrA$.  
We say that $\scrA$ is a $*$-algebra if there is a conjugate linear map $(\cdot)^*:\scrA\to\scrA$ such that $(a^{*})^* = a$ for all $a\in\scrA$.
Given a $*$-algebra $\scrA$, we say that $\scrA$ is a Banach algebra if there is a norm $\banachnorm{\cdot}{\scrA}$ on $\scrA$ making $\scrA$ into a Banach space such that
\begin{equation}
    \text{$\banachnorm{a}{\scrA} = \banachnorm{a^*}{\scrA}$ and $\banachnorm{ab}{\scrA}\leq \banachnorm{a}{\scrA}\banachnorm{b}{\scrA}$}
\end{equation}
for all $a, b\in\scrA$. 
We call a Banach algebra $\scrA$ a $C^*$-algebra if $\banachnorm{a^*a}{\scrA} = \banachnorm{a}{\scrA}^2$ holds for all $a\in\scrA$. 
We say that a $C^*$-algebra $\scrA$ is a von Neumann algebra if there is a Hilbert space $\scrH$ such that $\scrA\subseteq\bops{\scrH}$ in which $\scrA$ is WOT-closed. 
There is a standard alternative description of von Neumann algebras of which we shall make free use is the following: 
a $C^*$-algebra $\scrA$ is a von Neumann algebra if and only if there exists a Banach space $\predual{\scrA}$ such that $\dual{\seq{\predual{\scrA}}}$ is isometrically isomorphic to $\scrA$. 
Given a $C^*$-algebra $\scrA$, we let $\positives{\scrA}$ denote the set of elements in $\scrA$ given by those $b\in\scrA$ for which there exists $a\in\scrA$ with $b = a^*a$; it is straightforward to check that $\positives{\scrA}$ is a positive cone in $\scrA$. 
We write $a\geq 0$ to denote that $a\in\positives{\scrA}$, and we say $a\in\positives{\scrA}$ is positive.
For any $a\in\positives{\scrA}$, there exists a unique $b\in\positives{\scrA}$ such that $b^2 = a$; we write $\sqrt{b}$ or $b^{1/2}$ to denote this element. 
For any $a\in\scrA$, $a^*a\in\positives{\scrA}$, and we write $|a|$ to denote $\sqrt{a^*a}$.
If $a\geq 0$ is invertible, we call $a$ strictly positive and write $a > 0$. 
We write $\selfadj{\scrA}$ to denote the set of elements $a\in\scrA$ with $a^* = a$, and we call such $a$ self-adjoint; note that $\positives{\scrA}\subset\selfadj{\scrA}$.
We say $p\in\scrA$ is a projection if $p^* = p = p^2$, and, if $\scrA$ is unital with unit $1_\scrA$, we call $u\in\scrA$ unitary if $u^*u = uu^* = 1_\scrA$. 
A bounded linear functional $\varphi\in\dual{\scrA}$ is called a state if $\varphi(a)\geq 0$ for all $a\in\positives{\scrA}$ and $\varphi(1_\scrA) = 1$. 
% %
A state $\tau\in\dual{\scrA}$ is called a trace if $\tau(ab) = \tau(ba)$ for all $a, b\in\scrA$, and $\tau$ is called faithful if $\ker\tau\cap\positives{\scrA} = \set{0}$.
% %
% We call a von Neumann algebra $\scrM$ tracial if there exists a faithful trace $\tau\in\dual{\scrM}$.
%

%
Let $\scrA$ be a $C^*$-algebra, and for $d\in\mbN$ let $\matrices$ denote the $d\times d$ matrices with entries in $\mbC$. 
Note that $\matrices\cong\bops{\mbC^d}$ is itself a $C^*$-algebra.
Make the algebraic tensor product $\scrA\otimes\matrices$ into a $*$-algebra by extending $(a\otimes b)^* = a^*\otimes b^*$ linearly, and make this $*$-algebra into a $C^*$-algebra by letting $\banachnorm{\cdot}{\scrA\otimes\matrices}$ be the unique norm making $\scrA\otimes\matrices$ into a $C^*$-algebra \cite[Theorem 6.3.9]{Murphy2007C-AlgebrasTheory}.
A map $\psi\in\linears{\scrA}$ is called positive if $\psi\seq{\positives{\scrA}}\subseteq\positives{\scrA}$, and $\psi$ is called completely positive if for all $d\in\mbN$ the map $\psi\otimes\identityop{\matrices}\in\linears{\scrA\otimes\matrices}$ is positive, where $\identityop{\matrices}\in\linears{\matrices}$ is the identity. 
Any positive map $\psi\in\linears{\scrA}$ is a $*$-map, meaning $\psi(a^*) = \psi(a)^*$ for all $a\in\scrA$. 
Call $\psi\in\linears{\scrA}$ unital if $\psi(1_\scrA) = 1_\scrA$, and if $\tau$ is a trace on $\scrA$, call say that $\psi$ is $\tau$-preserving if $\tau\circ\psi = \tau$. 
From \cite{Choi1974AC-Algebras}, any completely positive unital map $\psi\in\bops{\scrA}$ satisfies the Schwarz inequality 
\begin{equation}\label{Eqn:Schwarz_inequality}
    \psi(a^*a)\geq \psi(a)^*\psi(a)
\end{equation}
for all $a\in\scrA$. 
In general, we call a map satisfying (\ref{Eqn:Schwarz_inequality}) a Schwarz map. 
For facts about $C^*$-algebras, von Neumann algebras, and completely positive maps, one may consult the standard references \cite{Murphy2007C-AlgebrasTheory}, \cite{Takesaki1979TheoryI, Takesaki2003TheoryII}, and \cite{Paulsen2003CompletelyAlgebras}, respectively.
In this work, we are primarily concerned with various Banach subspaces, Hilbert subspaces, and von Neumann subalgebras of the vector space of random matrices. 
Towards this end, let us recall the relevant matrix-analytic notions and set terminology. 
We let $\mbI = 1_{\matrices}$ denote the identity matrix in $\matrices$. 
For any $a\in\matrices$, let $\proj{a}$ denote the projection onto $\ran(a)$.
For $a\in\matrices$, let $\tr{a}$ denote the sum of the eigenvalues of $a$, and let $\infnorm{a}$ be the operator norm on $a$ by viewing $\matrices$ as $\bops{\mbC^d}$ when $\mbC^d$ is given the standard Euclidean Hilbert space structure.
Let $\tracenorm{\cdot}$ denote the trace norm $\tracenorm{\cdot} = \tr{|\cdot|}$. 
We let $\innerHS{\cdot}{\cdot}$ denote the Hilbert-Schmidt inner product on $\matrices$
\begin{equation}
    \innerHS{a}{b}
    =
    \tr{a^*b}\qquad a, b\in\matrices.
\end{equation}
We write $\hsnorm{\cdot}$ to denote the norm on $\matrices$ corresponding to this inner product. 
For a map $\psi\in\linears{\matrices}$, let $\|\psi\|_{\#}$ denote the operator norm of $\psi$ with respect to the $\|\cdot\|_\#$ norm on $\matrices$, where $\#\in\set{1, \infty, \operatorname{HS}}$.
We let $\states$ denote the set 
\begin{equation}
    \states 
    =
    \set{
    \rho\in\positives{\matrices}\,\,:\,\,\tr{\rho} = 1
    }.
\end{equation}
By the Riesz representation theorem applied to $\seq{\matrices, \innerHS{\cdot}{\cdot}}$, it is easy to see that $\states$ is in bijective correspondence with the set of states in $\dual{\matrices}$; for this reason we call $\rho\in\states$ a state. 
Note that $\linears{\matrices} = \bops{\matrices}$. 
Given $\psi\in\linears{\matrices}$, we write $\psi^\dagger$ to denote the adjoint of $\psi$ with respect to the inner product $\innerHS{\cdot}{\cdot}$.
We call a map $\phi\in\linears{\matrices}$ a quantum channel if $\phi$ is completely positive and $\operatorname{Tr}{\cdot}$-preserving, and we denote the set of all such maps by $\channels$.
From the Riesz representation theorem, $\phi\in\channels$ if and only if $\phi^\dagger$ is completely positive and unital; in particular, the adjoint of a quantum channel satisfies the Schwarz inequality (\ref{Eqn:Schwarz_inequality}), a fact we use regularly in the following. 
Fix $\seq{\Omega, \mcF, \mu}$ a standard probability space. 
For a measurable function $f:\Omega\to\mbC$, let $\avg{f}$ denote the quantity 
\begin{equation}
    \avg{f}
        =
    \int_\Omega 
    f(\omega)\,\dee\mu(\omega)
\end{equation}
whenever it is defined.
Let $\rmatrices$ denote the set of random  $d\times d$ matrices, 
\begin{equation}
    \rmatrices 
        =
    \set{a:\Omega\to\matrices\text{ measurable}}/\sim,
\end{equation}
where $\sim$ denotes almost everywhere equivalence. 
We let $\positives{\rmatrices}$ (resp. $\selfadj{\rmatrices}$) denote the subset of $\rmatrices$ consisting of $a\in\rmatrices$ with $a\in \positives{\matrices}$ (resp. $a\in\selfadj{\matrices}$) almost surely and we call such a random matrix positive (resp. self-adjoint).
For $a\in\positives{\rmatrices}$ with $a>0$ almost surely, we write $a>0$. 
We call a random matrix $a\in\rmatrices$ a (random) projection (resp. unitary) if $a$ is a projection (resp. unitary) almost surely. 
Let $\rstates$ denote the subset of $\positives{\rmatrices}$ of all $\rho\in\positives{\rmatrices}$ with $\tr{\rho} = 1$ almost surely. 
For $a\in\rmatrices$, we often implicitly fix a representative of $a$ and write $a_\omega$ to denote the evaluation of this representative at a point $\omega\in\Omega$. 
Note that $\rmatrices$ is a $*$-algebra with pointwise almost everywhere operations. 
We identify $\matrices$ with the subalgebra of $\rmatrices$ consisting of $a'\in\rmatrices$ for which there is $a\in\matrices$ with $a' = a$ almost surely.
For $a\in\rmatrices$, let $\linfnorm{a}$ be the norm 
\begin{equation}
    \linfnorm{a}
        =
    \operatorname{ess-sup}_{\omega\in\Omega}\infnorm{a_\omega}
\end{equation}
Then let $\Linfty{\Omega}{\matrices}$ denote the von Neumann subalgebra of $\rmatrices$ defined by 
\begin{equation}
    \Linfty{\Omega}{\matrices}
    =
    \set{a\in\rmatrices\,\,:\,\, \linfnorm{a}<\infty}.
\end{equation}
It is not hard to see that $\Linfty{\Omega}{\matrices}\cong L^\infty(\Omega)\otimes\matrices$, and that $ a\mapsto 
    \avg{\trlower{a}}$
defines a faithful trace on $\Linfty{\Omega}{\matrices}$.
We denote this trace by $\Linftrace{\cdot}$. 
It is a standard fact that $\predual{L^\infty(\Omega)}\cong L^1(\Omega)$, hence $\predual{\Linfty{\Omega}{\matrices}}\cong \Lone{\Omega}{\matrices}$, where 
\begin{equation}
    \Lone{\Omega}{\matrices}
    =
    \set{a\in\rmatrices
        \,\,:\,\,
        \Linftrace{|a|}<\infty
        }.
\end{equation}
For $a, b\in\rmatrices$, let $\Ltwozinner{a}{b}{\mbI}$ denote 
\begin{equation}
    \Ltwozinner{a}{b}{\mbI}
    =
    \int_\Omega 
    \innerHS{a_\omega}{b_\omega}\,\dee\mu(\omega)
\end{equation}
(whenever this quantity is defined), and let $L^2(\mbI)$ be the subset of $\rmatrices$ given by 
\begin{equation}
    L^2(\mbI)
    =
    \set{a\in\rmatrices
        \,\,:\,\,
         \Ltwozinner{a}{a}{\mbI}<\infty
    }.
\end{equation}
We see that $L^2(\mbI)$ is a Hilbert space with $\Ltwozinner{\cdot}{\cdot}{\mbI}.$
\subsection{The operator-theoretic perspective on disordered quantum processes}
Recall from the introduction that we said the disordered quantum process defined by $\seq{\theta, \phi}$ was reduced by a projection $p\in\rmatrices$ if 
\begin{equation}
    \phi_{\theta(\omega)}\!\seq{p_\omega\matrices p_\omega}
    \subseteq p_{\theta(\omega)}\matrices p_{\theta(\omega)}
\end{equation}
held almost surely. 
We now rephrase this condition in terms of certain operators on $\rmatrices$, which is a useful perspective enabling functional analytic techniques. 
We say a map $\theta:\Omega\to\Omega$ is a measure preserving transformation (\textit{mpt}) if for all events $F\in\mcF$, we have that $\prob{\theta^{-1}(F)} = \prob{F}$.  
We call $\theta$ invertible if $\theta^{-1}$ exists and is measurable. 
We call $\theta$ ergodic if the hypothesis 
\begin{equation}
\text{for all events $F\in\mcF$, }
    \prob{\theta^{-1}(F)\bigtriangleup F}
    =
    0
    \text{ implies that }
    \prob{F}\in\set{0, 1}
\end{equation}
is satisfied, where $\bigtriangleup$ denotes the symmetric difference. 
For a measurable map $T:\Omega\to\Omega$ and a measurable function $f:\Omega\to\mbC$, we let $\koopman{T}(f)$ denote the function $f\circ T$. 
It is clear that whenever $T$ is an invertible mpt, $\koopman{T}$ defines a unitary operator in $\bops{L^2(\Omega)}$. 
%
%Whenever we speak of $\sigma(\koopman{T})$, we shall be referring to the spectrum of $\koopman{T}$ when $\koopman{T}$ is viewed as an element of $\bops{L^2(\Omega)}$.
%
It is a standard fact that an invertible mpt $T$ is ergodic if and only if the eigenvalue of $\koopman{T}$ at $\lambda = 1$ is nondegenerate. 
Let $\mbT = \set{\lambda\in\mbC\,\,:\,\, |\lambda| = 1}$. 
It is known \cite{Strmer1974SpectraTransformations} that $\sigma(\koopman{T}) = \mbT$ whenever $T$ is an ergodic invertible mpt and we view $\koopman{T}$ as an operator on $L^2(\Omega)$. 
Let us now recall some facts about ergodic, invertible mpt that we will make use of in the following. 
\begin{prop}\label{Prop:Basics_of_ergodicity}
     Let $\theta:\Omega\to\Omega$ be an invertible mpt.
    The following are equivalent. 
    \begin{enumerate}[label = (\alph*)]
        \item $\theta$ is ergodic.

        \item $\theta^{-1}$ is ergodic.

        \item For any $E\in\mcF$ with $\theta^{-1}(E)\subset E$, $\mu[E]\in\set{0, 1}$.

        \item For any measurable function $f:\Omega\to\mbR$, the almost sure equality $\koopman{\theta}(f) = f$ implies that $f$ is constant. 
    \end{enumerate}
\end{prop}
\begin{proof}
    See \cite[Ch. 4]{Viana2015FoundationsTheory} or \cite[Ch. 1]{Walters1982AnTheory}.
\end{proof}
From this, we can prove the following standard fact of which we make use at multiple times in our technical arguments. 
\begin{lemma}\label{Lem:Increasing_implies_constant}
    Assume $\theta:\Omega\to\Omega$ is an ergodic mpt. 
    If $f:\Omega\to\mbR$ satisfies $\koopman{\theta}(f)\leq f$ almost surely, then there is $C\in\mbR$ with $f = C$ almost surely. 
\end{lemma}
\begin{proof}
    Let $\varepsilon>0$ and let $c\in\mbR$ be such that $\prob{f\in (c - \varepsilon, c + \varepsilon)} > 0$. 
    By Poincaré recurrence, for almost every $\omega\in\Omega$ with $f(\omega)\in (c - \varepsilon, c + \varepsilon)$, there is $n\in\mbN$ such that $f(\theta^n(\omega))\in (c - \varepsilon, c + \varepsilon)$.
    Because $\koopman{\theta}(f)\leq f$ almost surely, we have that 
    \begin{equation}
        c - \varepsilon < f(\theta^n(\omega))\leq\cdots \leq f(\theta(\omega))\leq f(\omega) < c + \varepsilon,
    \end{equation}
    so in particular $f(\theta(\omega))\in (c - \varepsilon, c + \varepsilon)$. 
    So, by Proposition \ref{Prop:Basics_of_ergodicity}, we conclude that $\prob{f\in (c - \varepsilon, c + \varepsilon)} = 1$. 
    Therefore, if $f$ were nonconstant, there would be distinct $c_1, c_2\in\mbR$ and $\varepsilon>0$ small enough so that $ (c_1 - \varepsilon, c_1 + \varepsilon)\cap  (c_2 - \varepsilon, c_2 + \varepsilon) = \emptyset$ and
    \begin{equation}
        \prob{f\in (c_1 - \varepsilon, c_1 + \varepsilon)}>0\quad\text{and}\quad
        \prob{f\in (c_2 - \varepsilon, c_2 + \varepsilon)}>0,
    \end{equation}
    which, by the above argument, implies $ \prob{f\in (c_1 - \varepsilon, c_1 + \varepsilon)}=
        \prob{f\in (c_2 - \varepsilon, c_2 + \varepsilon)} = 1$, an absurdity. 
    Therefore, it must be that $f$ is almost surely constant, which ends the proof. 
\end{proof}
Now we may introduce the operator-theoretic perspective on ergodic quantum processes. 
For a measurable $\psi:\Omega\to\linears{\matrices}$ and $T:\Omega\to\Omega$, we let $\glop{T}{\psi}\in\linears{\rmatrices}$ denote the linear operator defined by 
\begin{equation}
    \begin{split}
        \glop{T}{\psi}:\rmatrices&\to\rmatrices\\
        \glop{T}{\psi}(a)_\omega &:= \psi_{\omega}\!\seq{a_{T(\omega)}}.
    \end{split}
\end{equation}
Phrased differently, if we let $M_\psi\in\linears{\rmatrices}$ be the operator $M_\psi(a)_\omega = \psi_\omega(a_\omega)$, and if we identify $\koopman{T}$ with the linear operator in $\linears{\rmatrices}$ defined by $\koopman{T}(a)_\omega = a_{T(\omega)}$, then we see that 
\begin{equation}
    \glop{T}{\psi}
    =
    M_\psi\circ\koopman{T}.
\end{equation}
If $T:\Omega\to\Omega$ is invertible, $\glopdag{T}{\psi}$ to denote $\glop{T^{-1}}{\koopman{T^{-1}}(\psi^\dagger)}$. 
That is, $\glopdag{T}{\psi}$ is the map 
\begin{equation}
    \glopdag{T}{\psi}(a)_\omega
        =
    \psi^\dagger_{T^{-1}(\omega)}\seq{a_{T^{-1}(\omega)}}.
\end{equation}
We may now rephrase reducibility in terms of $\glop{T}{\psi}$.
\begin{definition}[Reducing projections]\label{Defn:Reducing_projection_for_glop}
    Let $T:\Omega\to\Omega$ be measurable, fix $\psi:\Omega\to\linears{\matrices}$, and let $p\in\rmatrices$ be a projection
    We say that $p$ reduces $\glop{T}{\psi}$ if for all $a\in\rmatrices$ with $pap = a$, we have that $p\glop{T}{\psi}(a)p = \glop{T}{\psi}(a)$, and we call such $p$ a reducing projection. 
    If the only projection reducing $\glop{T}{\psi}$ is $p = \mbI$, we call $\glop{T}{\psi}$ irreducible. 
\end{definition}
\begin{remark}
    The relationship between this version of reducibility and the one given in the introduction is the following: 
    if $T$ is invertible, then notice that $\glop{T^{-1}}{\psi}$ is reduced by $p$ is equivalent to 
    \begin{equation}
        \koopman{T}\!\seq{\glop{T^{-1}}{\psi}(pap)}
        =
        \koopman{T}(p)\koopman{T}\!\seq{\glop{T^{-1}}{\psi}(a)}\koopman{T}(p)
    \end{equation}
    for all $a\in\rmatrices$, which, evaluated pointwise, becomes 
    \begin{equation}
        \psi_{T(\omega)}\!\seq{p_\omega a_\omega p_\omega}
        =
        p_{T(\omega)}\psi_{T(\omega)}(a_\omega) p_{T(\omega)},
    \end{equation}
    i.e., $\psi_{T(\omega)}\!\seq{p_\omega \matrices p_\omega}\subseteq p_{T(\omega)}\matrices p_{T(\omega)}$.
\end{remark}
In \cite{Ekblad2024ReducibilityProcesses}, the following characterization of reducibility was given.
\begin{prop}[\texorpdfstring{\cite{Ekblad2024ReducibilityProcesses}}{l}]\label{Prop:Characterization_of_reducing_proj}
    Let $T:\Omega\to\Omega$ be measurable and assume $\psi:\Omega\to\linears{\matrices}$ is positive almost surely.
    Let $p\in\rmatrices$ be a projection. 
    The following are equivalent. 
    \begin{enumerate}[label = (\alph*)]
        \item $p$ reduces $\glop{T}{\psi}$. 

        \item $p\glop{T}{\psi}(p)p = \glop{T}{\psi}(p)$.

        \item There is $f:\Omega\to (0, \infty)$ such that $\glop{T}{\psi}(p)\leq f p$ almost surely.
    \end{enumerate}
    If we assume that there is a constant $C\in (0, \infty)$ such that $\|\glop{T}{\psi}(x)\|_\infty\leq C \|x\|_\infty$ almost surely for all $x\in\rmatrices$, then the above conditions are equivalent to 
    \begin{enumerate}[label = (d)]
        \item  $\glop{T}{\psi}(p)\leq C p$ almost surely. 
    \end{enumerate}
    Moreover, if $T$ is an invertible mpt, $p$ reduces $\glop{T}{\psi}$ if and only if $\mbI - p$ reduces $\glopdag{T}{\psi}$.
\end{prop}
In this work, we shall study $\glop{T}{\psi}$ when $T$ is an invertible ergodic mpt and $\psi$ is a random quantum channel. 
Towards this end, let us fix a distinguished invertible ergodic mpt $\theta:\Omega\to\Omega$ and random quantum channel 
\begin{equation}
    \phi:\Omega\to\channels,
\end{equation}
and let $\opL\in\linears{\rmatrices}$ denote the linear operator
\begin{equation}
    \opL 
    =
    \glop{\theta^{-1}}{\phi}.
\end{equation}
That is, $\opL(a)_\omega = \phi_\omega(a_{\theta^{-1}(\omega)})$. 
Let $\opLdag$ denote $\glop{\theta^{-1}}{\phi}^\dagger = \glop{\theta}{\koopman{\theta}(\phi^\dagger)}$.
In \cite{Ekblad2024ReducibilityProcesses}, the theory of the maps $\opL$ and $\opLdag$ was studied and developed, where the main result was the characterization (Theorem \ref{Thmx:irreducibility}) given in the introduction. 
Let us rephrase that characterization (and state the stronger version of the theorem proved in \cite{Ekblad2024ReducibilityProcesses}) in terms of $\opL$ now. 
Let $\FixL$ denote the set of random matrices $a\in\rmatrices$ such that $\opL(a) = a$. 
\begin{thmx}[\texorpdfstring{\cite{Ekblad2024ReducibilityProcesses}}{l}]
\label{Thmx:Irreducibility_classification_operator}
    The following are equivalent. 
    \begin{enumerate}[label = (\alph*)]
        \item $\opL$ is irreducible. 

        \item $\FixL$ is a one-dimensional subspace of $\rmatrices$. 

        \item There exists a unique $\varrho\in\rstates$ such that $\opL(\varrho) = \varrho$, and this $\varrho$ satisfies $\varrho>0$.
    \end{enumerate}
\end{thmx}
In this work, we are concerned with the situation where $\opL$ is irreducible, and the central computational object is the unique state $\varrho\in\rstates$ of the above theorem. 
This state is precisely the unique steady state of $\eqp$ described in the introduction, and we make frequent reference to it, so let us name it accordingly. 
\begin{definition}[Unique steady state of \texorpdfstring{$\opL$}{l}]
    For irreducible $\opL$, we call the unique state $\varrho\in\rstates$ of Theorem \ref{Thmx:Irreducibility_classification_operator} the \textit{unique steady state} of $\opL$. 
\end{definition}
We now turn towards studying the spectral theory of $\opL$ and $\opLdag$. 
\subsection{Basic properties of eigenspectrum}
Let us set notation and collect preliminary facts. 
Our basic goal here is to establish some of the basic group-theoretic properties of the eigenspectrum of $\eqp$ discussed in the introduction.
It is often necessary to specify which closed subspace of $\rmatrices$ we are viewing $\opL$ and $\opLdag$ as acting on, especially when discussing the spectrum.
Let $\opL_1\in\linears{L^1\seq{\Omega; \matrices}, \rmatrices}$ denote the map $\opL\vert_{L^1\seq{\Omega; \matrices}}$, and write $\opLdag_1$ to denote the restriction of $\opLdag$ to $\Linfty{\Omega}{\matrices}$.
\begin{lemma}\label{Lem:Operator_norms_of_opL_and_adjoint_formula}
The map $\opL_1$ defines an element of $\bops{L^1\!\seq{\Omega; \matrices}}$ satisfying $\banachnorm{\opL_1}{L^1\!\seq{\Omega; \matrices}} = 1$, and, under the duality relation $\dual{L^1\seq{\Omega; \matrices}} = \Linfty{\Omega}{\matrices}$, we have $\opLdag_1 = \dual{\opL_1}$.
%
% %
% %
\end{lemma}
\begin{proof}
    That $\banachnorm{\opL_1}{L^1(\Omega; \matrices)} = 1$ follows immediately from the fact that $\phi$ is almost surely trace-preserving taken with the Russo-Dye theorem \cite[Theorem 3.39]{Watrous2018TheInformation}, and it is immediate that  $\opLdag_1 = \dual{\opL_1}$.
\end{proof}
Now, let $\mbT = \set{\lambda\in\mbC\,\,:\,\, |\lambda| = 1}$. 
From this lemma and Theorem \ref{Thmx:Irreducibility_classification_operator}, we know that $\opL_1$ always has a fixed point, so in particular $\specrad{\opL} = 1$ hence $\perspec{\opL_1}\subseteq\mbT$. 
We write $\Lambda_\opL$ to denote $\perspec{\opL_1}\cap\eigspec{\opL_1}$, we write $\Lambda_{\opL}^\dagger$ to denote $\perspec{\opLdag_1}\cap\eigspec{\opLdag_1}$, and we write $\Lambda_\theta$ to denote $\eigspec{\koopman{\theta}\vert_{L^1(\Omega)}}$. 
The next lemma shows that we need not worry about the domain of $\opL$ and $\opLdag$ when considering eigenspectrum, i.e., the eigenspectrum of $\opL$ and $\opLdag$ when restricted to $\Linfty{\Omega}{\matrices}$ (or $\Lone{\Omega}{\matrices}$) coincides with the eigenspectrum of $\opL$ and $\opLdag$ as operators on $\rmatrices$.
\begin{lemma}
    For any $\alpha\in\mbC$ and $x\in\rmatrices$, if either $\opL(x) = \alpha x$ or $\opLdag(x) = \alpha x$, then $x\in\Linfty{\Omega}{\matrices}$. 
\end{lemma}
\begin{proof}
    First, we show that any such $\alpha$ satisfies $|\alpha|\leq 1$.
    For contradiction, suppose that $x\in\rmatrices$ is nonzero and satisfies $\opL(x) = \alpha x$ almost surely with $|\alpha| > 1$. 
    By the Russo-Dye theorem \cite[Theorem 3.39]{Watrous2018TheInformation}, we know that $\onenorm{\phi_\omega} = 1$ almost surely, as $\phi$ is trace preserving almost surely. 
    Therefore, from $\opL(x) = \alpha x$, we discover that 
    \begin{equation}\label{Eqn:alpha_leq_1}
        |\alpha|\onenorm{x_{\theta(\omega)}}
        =
        \onenorm{\phi_{\theta(\omega)}(x_{\omega})}
        \leq 
        \onenorm{x_\omega}
    \end{equation}
    holds almost surely. 
    Now let $\varepsilon>0$ be such that $\prob{\onenorm{x} \geq \varepsilon} > 0$.
    Then for almost every $\omega$ for which $\onenorm{x_\omega}\geq \varepsilon$, Poincaré recurrence implies that there exists a strictly increasing sequence $\seq{n_k}_{k\in\mbN}\subset\mbN$ with $\varepsilon
        \leq \onenorm{x_{\theta^{n_k}(\omega)}}$ for all $k$. 
    But from (\ref{Eqn:alpha_leq_1}), we know that 
    \begin{equation}
        \varepsilon
        \leq 
        \onenorm{x_{\theta^{n_k}(\omega)}}
        \leq 
        |\alpha|^{-n_k}
        \onenorm{x_\omega}
    \end{equation}
    for almost every $\omega$ with $\onenorm{x_\omega}\geq \varepsilon$ and all $k$. 
    Since $n_k\to\infty$ as $k\to\infty$, however, this implies $\varepsilon = 0$, a contradiction. 
    Therefore, $|\alpha|\leq 1$. 
    Arguing similarly in the case that $x\in\rmatrices$ satisfies $\opLdag(x) = \alpha x$, we may apply the Russo-Dye theorem once more to find that 
    \begin{equation}\label{Eqn:alpha_leq_1_eqn_2}
        |\alpha|\infnorm{x_{\omega}}
        =
        \infnorm{\phi_{\theta(\omega)}^\dagger(x_{\theta(\omega)})}
        \leq 
        \infnorm{x_{\theta(\omega)}}
    \end{equation}
    holds almost surely, so we conclude once more from this that $|\alpha|\leq 1$, as claimed. 
    Now, if $\opL(x) = \alpha x$, then since $\phi$ is trace preserving almost surely, we know that $|\tr{\koopman{\theta^{-1}}(x)}| = |\alpha| |\tr{x}|\leq |\tr{x}|$. 
    Thus, by Lemma \ref{Lem:Increasing_implies_constant}, we conclude that $|\tr{x}| = C$ almost surely for some deterministic $C\in (0, \infty)$. 
    So, because there is a universal constant $D\in (0, \infty)$ such that $D|\tr{\cdot}|\geq \infnorm{\cdot}$, we conclude that $\infnorm{x}\leq CD$, i.e., $x\in\Linfty{\Omega}{\matrices}$.
    On the other hand, if $\opLdag(x) = \alpha x$, then $\infnorm{x} = \infnorm{\opLdag(x)}\leq \infnorm{\koopman{\theta}(x)}$, so, by Lemma \ref{Lem:Increasing_implies_constant}, ergodicity of $\theta$ implies there is a deterministic constant $C$ with $\infnorm{x} = C$ almost surely, which immediately implies $x\in\Linfty{\Omega}{\matrices}$.
\end{proof}
\begin{lemma}\label{Lem:Eigenspectrum_koopman_in_opL}
    $\Lambda_\theta\subseteq \Lambda_\opL\cap\Lambda_\opL^\dagger$. 
\end{lemma}
\begin{proof}
    For any nonzero $f\in L^1(\Omega)$ with $\koopman{\theta}(f) = \alpha f$, we have that $|\alpha||f| = |\koopman{\theta}(f)|$, hence $|\alpha|^n|f| = |\koopman{\theta}^n(f)|$. 
    Therefore, by the pointwise ergodic theorem, the almost sure equality
    \begin{equation}
        \avg{|f|} = \lim_{N}
        \frac{1}{N}
        \sum_{n=1}^N 
        |\koopman{\theta}^n(f)|
        =
        |f|
        \lim_{N}
        \frac{1}{N}
        \sum_{n=1}^N 
        |\alpha|^n
    \end{equation}
    implies that $|\alpha| = 1$. 
    In particular, $|f| = |\koopman{\theta}(f)|$, so since $\theta$ is ergodic, $|f| = \avg{|f|}$ almost surely. 
    Because $f$ is nonzero, we may assume without loss of generality that $|f| = 1$ almost surely. 
    Now, it is clear that 
    \begin{equation}
        \opLdag(f\mbI)
        =
        \alpha f \mbI,
    \end{equation}
    so $\Lambda_\theta\subseteq\Lambda_\opL^\dagger$ already. 
    On the other hand, notice that because $\theta$ is invertible, $\koopman{\theta}(f) = \alpha f$ is equivalent to $\overline{\alpha} f = \koopman{\theta^{-1}}(f)$. 
    So, if $\varrho$ is the unique steady state of $\opL$, then we see that 
    \begin{equation}
        \opL(f^{-1}\varrho)
        =
        \overline{\alpha}f\varrho.
    \end{equation}
    But this held for arbitrary $\alpha\in\Lambda_\theta$, so in particular if we notice that $\koopman{\theta}(f^{-1}) = \overline{\alpha}f^{-1}$, then replacing $\alpha$ with $\overline{\alpha}$ in the above, we discover that $\alpha\in\Lambda_\opL$, ending the proof. 
\end{proof}
One of the main attributes of complete positivity is that completely positive and unital maps satisfy the Schwarz inequality. 
Because $\phi^\dagger$ is completely positive and unital almost surely, $\opLdag$ satisfies the Schwarz inequality 
\begin{equation}
    \opLdag(a^*a)
        \geq
    \opLdag(a)^*\opLdag(a)
\end{equation}
for all $a\in \rmatrices$. 
This has various implications for the spectral theory of $\opL$ and $\opLdag$, as we shall now see. 
Most pertinent to this is the general work of Groh on unital Schwarz maps \cite{Groh1981TheC-Algebras, Groh1983OnC-Algebras, Groh1984UniformW-Algebras, Groh1984UniformlyC-algebras}, which, although too general to give many specific results for $\opLdag$, nevertheless gives us a basic framework to start from.
Let us collect these facts and some easy corollaries as they apply to $\opLdag$ in a single proposition.
\begin{prop}\label{Prop:Polar_decomp}
    Assume $\opL$ is irreducible with unique steady state $\varrho$. 
    Let $\alpha\in\mbT$.
    \begin{enumerate}[label = (\alph*)]
        \item $\operatorname{dim}\set{x\in\rmatrices\,\,:\,\, \opL(x) = \alpha x} \leq 1$.
        In particular, all $\alpha\in\Lambda_\opL$ are simple eigenvalues. 
        
        \item For any $x\in\rmatrices$ with $\opL(x) = \alpha x$, there is a unique constant $\lambda\in\mbC$ and a unique unitary $u\in\rmatrices$ such that $x = \lambda u \varrho$.
    \end{enumerate}
\end{prop}
\begin{proof}
    The main thing to note here is that $\opL_1$ is the preadjoint of the unital Schwarz map $\opLdag_1$ on the von Neumann algebra $\Linfty{\Omega}{\matrices}$, which is precisely the scenario in which the results of Groh apply. 
    With this in mind, let us prove (a).
    First, notice that if $\opL(x) = \alpha x$, then because $\phi$ is trace preserving almost surely, we have that $|\tr{\koopman{\theta^{-1}}(x)}| = |\tr{x}|$ almost surely. 
    So, since $\theta^{-1}$ is ergodic, this implies $|\tr{x}| = C$ for some constant $C<\infty$. 
    Thus, $x\in\Lone{\Omega}{\matrices}$. 
    Therefore, $\set{x\in\rmatrices\,\,:\,\, \opL(x) = \alpha x}\subset \Lone{\Omega}{\matrices}$, so (a) follows directly from \cite[Theorem 2.2]{Groh1984UniformlyC-algebras} upon noting that, since $\opL$ is irreducible, Theorem \ref{Thmx:Irreducibility_classification_operator} implies that $\dim\ker\seq{1 - \opL_1} = 1$. 
    Next, to see (b), from \cite[Proposition 2.3 (a)]{Groh1984UniformW-Algebras} we have that $\opL(|x|) = |x|$ for any $x$ as in (b).
    Therefore, by irreducibility, $|x| = \lambda \varrho$ for some $\lambda\in\mbC$. 
    Because $\varrho>0$, we know that $x$ is almost surely invertible. 
    In particular, the polar decomposition $x = u\varrho$ defines a unique unitary $u\in\rmatrices$, as desired. 
\end{proof}
\begin{remark}
    The spectral theory of a certain sort of unital Schwarz maps on von Neumann algebras satisfying a certain sort of irreducibility was studied extensively by Groh \cite{Groh1981TheC-Algebras, Groh1983OnC-Algebras, Groh1984UniformW-Algebras, Groh1984UniformlyC-algebras}. 
    However, the sort of irreducibility considered there is far too restrictive for our purposes: indeed, the spectral results for irreducible maps considered by Groh apply to $\opL^\dagger$ only if $\theta:\Omega\to\Omega$ is a cyclic permutation of a finite set. 
    See Appendix \ref{App:Groh irred} for more details.
\end{remark}
Next, we focus our attention on $\opLdag$. 
As mentioned above, $\opLdag_1$ is a unital Schwarz map. 
A basic object associated to such a map is the \textit{multiplicative domain}.
Recall that the multiplicative domain $\mcM_\psi$ of a linear map $\psi:\mcA\to\mcA$ of a $C^*$-algebra $\mcA$ is the set 
\begin{equation}
    \mcM_{\psi}
        :=
    \set{a\in\mcA\,\,:\,\,\psi(ab) = \psi(a)\psi(b)\text{ and }\psi(ba) = \psi(b)\psi(a)\text{ for all }b\in\mcA}.
\end{equation}
Under the assumption that $\psi$ is a unital Schwarz map, the well-known theorem of Choi \cite[Theorem 3.1]{Choi1974AC-Algebras} gives that 
\begin{equation}
    \multdom{\psi}
        =
    \set{a\in\scrA\,\,:\,\,
    \psi(a^*a) = \psi(a)^*\psi(a)
    \text{ and }
    \psi(aa^*) = \psi(a)\psi(a)^*}.
\end{equation}
This fact will prove useful at multiple points during the rest of this paper. 
We shall apply this theorem to the unital Schwarz maps $\opLdag:\Linfty{\Omega}{\matrices}\to\Linfty{\Omega}{\matrices}$ and $\phi_\omega^\dagger:\matrices\to\matrices$, and in particular we will make note of the following fact that is clear from the definitions.
\begin{lemma}\label{Lem:Mult_dom_direct_integral}
    $\multdom{\opLdag} = \set{x\in\rmatrices\,\,:\,\, x_\omega\in\multdom{\phi^\dagger_\omega}\text{ for a.e. $\omega$}}$.
\end{lemma}
Now, let us prove the following useful fact.
\begin{lemma}\label{Lem:Eigenmat_of_opLdag_in_multdom}
    Assume $\opL$ is irreducible and let $\alpha\in\mbT.$
    If $x\in\rmatrices$ satisfies $\opLdag(x) = \alpha x$, then $x\in\multdom{\opLdag}$.
\end{lemma}
\begin{proof}
    If $x\in\rmatrices$ satisfies $\opLdag(x) = \alpha x$, then notice that $\opLdag(x)^*\opLdag(x) = x^*x$, since $|\alpha| = 1$. 
    Therefore, by the Schwarz inequality, we have that $\opLdag(x^*x)\geq x^*x$. 
    Thus, if we let $\varrho$ be the unique steady state of $\opL$, then from
    \begin{equation}
        \avg{\tr{\varrho\opLdag(x^*x)}}
        =
        \avg{\tr{\opL(\varrho) x^*x}}
        =
        \avg{\tr{\varrho x^*x}}
    \end{equation}
    we conclude that $\avg{\tr{\varrho\seq{\opLdag(x^*x) - x^*x}}} = 0$. 
    Since $\varrho>0$ and $\opLdag(x^*x) - x^*x\geq 0$, this implies that $\opLdag(x^*x) = x^*x$. 
    Arguing similarly, we find that $\opLdag(xx^*) = xx^*$, so $x\in\multdom{\opLdag}$.
\end{proof}
This discussion has the following useful corollary. 
Let $\operatorname{Fix}(\opLdag) = \set{x\in\rmatrices\,\,:\,\,\opLdag(x) = x}$.
\begin{prop}\label{Prop:Fix_opLdag_is_CI}
    Assume $\opL$ is irreducible. 
    Then $\operatorname{Fix}(\opLdag) = \mbC\mbI$.
\end{prop}
\begin{proof}
    By Lemma \ref{Lem:Eigenmat_of_opLdag_in_multdom}, we know that $\operatorname{Fix}(\opLdag)\subset \multdom{\opLdag}$.
    Now, let $a\in\matrices$ be arbitrary. 
    Then for almost every $\omega\in\Omega$, the fact that $\opLdag(x) = x$ and Lemma \ref{Lem:Mult_dom_direct_integral} together imply 
    \begin{align*}
        \tr{\phi_{\theta(\omega)}(x_\omega \varrho_\omega) a}
        =
        \tr{\phi_{\theta(\omega)}^\dagger(x_{\theta(\omega)})\varrho_\omega \phi_{\theta(\omega)}^\dagger(a)}
        &=
        \tr{\varrho_\omega \phi_{\theta(\omega)}^\dagger(a x_{\theta(\omega)})}
        \\
        &= 
        \tr{ \phi_{\theta(\omega)}(\varrho_\omega) a x_{\theta(\omega)}}
        \\
        &= 
        \tr{ x_{\theta(\omega)}\varrho_{\theta(\omega)} a}.
    \end{align*}
    That is, $\opL(x\varrho) = x\varrho$, since $a$ was arbitrary. 
    So, since $\FixL = \mbC\varrho$, there is $\lambda\in\mbC$ such that $x\varrho = \lambda\varrho$, but since $\varrho>0$, this implies $x = \lambda\mbI$, as desired. 
\end{proof}
\begin{cor}\label{Cor:Eigenmat_of_opLdag_unitary}
    Assume $\opL$ is irreducible and let $\alpha\in\mbT$. 
    If $x\in\rmatrices$ satisfies $\opLdag(x) = \alpha x$, then there is a unitary $u\in\rmatrices$ and a constant $\lambda\in\mbC$ with $x = \lambda u$. 
\end{cor}
\begin{proof}
    Assume $x\neq 0$. 
    From Lemma \ref{Lem:Eigenmat_of_opLdag_in_multdom}, we know that $\opLdag(x^*x) = x^*x$ and $\opLdag(xx^*) = xx^*$. 
    Therefore, by the previous proposition, we conclude that there is a constant $\lambda'\in\mbC$ such that $x^*x = \lambda'\mbI$ and $xx^* = \overline{\lambda'}\mbI$. 
    But $x^*x\geq 0$, so $\lambda'\geq 0$, hence $\overline{\lambda'} = \lambda'$. 
    Let $\lambda = \sqrt{\lambda'}$.
    Then $u := \lambda x$ is unitary and satisfies $x = \lambda^{-1} u$, which concludes the proof.
\end{proof}
With this in hand, we can now prove that $\Lambda_\opL^\dagger$ is a group.
\begin{prop}\label{Prop:Lambda_opL_dagger_is_a_group}
    Assume $\opL$ is irreducible. Then $\Lambda_\opL^\dagger$ forms a group under multiplication. 
    % %
\end{prop}
\begin{proof}
    Let $\alpha, \gamma\in\Lambda_\opL^\dagger$ and let $u, v\in\rmatrices$ be the unitaries guaranteed by Corollary \ref{Cor:Eigenmat_of_opLdag_unitary} such that $\opLdag(u) = \alpha u$ and $\opLdag(v) = \gamma v$. 
    By Lemma \ref{Lem:Eigenmat_of_opLdag_in_multdom}, we know that $\opLdag(uv) = \opLdag(u)\opLdag(v) = \alpha\gamma uv$, so since $u$ and $v$ are unitary, $uv\neq 0$ hence $\alpha\gamma\in\Lambda_\opL^\dagger$. 
    Lastly, since $\opLdag$ is positive, we know that $\opLdag(u^*) = \opLdag(u)^* = \overline{\alpha}u^*$, hence $\overline{\alpha} = \alpha^{-1}\in\Lambda_\opL^\dagger$, concluding the proof.
\end{proof}
We are now in a position to conclude that $\Lambda_\opL$ is a group, which will follow from the next proposition. 
In the following, we let $G_\opL$ denote the set 
\begin{equation}
    G_\opL
        :=
    \set{\text{unitary }u\in\rmatrices\,\,:\,\, \text{there is $\alpha\in\Lambda_\opL^\dagger$ with }\opLdag(u) = \alpha u},
\end{equation}
which, as we saw in the proof of Proposition \ref{Prop:Lambda_opL_dagger_is_a_group}, forms a group under matrix multiplication. 
\begin{prop}\label{Prop:Lambda_opL_is_Lambda_opL_dagger}
    Assume $\opL$ is irreducible with unique steady state $\varrho$.
    \begin{enumerate}[label = (\alph*)]
        \item For $\alpha\in\Lambda_\opL^\dagger$ and $u\in G_\opL$ with $\opLdag(u) =\alpha u$, $\opL(u\varrho) = \overline{\alpha} u\varrho$.

        \item For $\alpha\in\Lambda_\opL$ and $x\in\rmatrices$ with $\opL(x) = \alpha x$, then $\opLdag(u) = \overline{\alpha} u$, where $u$ is the unitary from Proposition \ref{Prop:Polar_decomp}.
    \end{enumerate}
    In particular, $\Lambda_\opL = \Lambda_\opL^\dagger$.
\end{prop}
\begin{proof}
    By Proposition \ref{Prop:Lambda_opL_dagger_is_a_group}, it is clear that (a) and (b) together imply $\Lambda_\opL = \Lambda_\opL^\dagger$, so it remains to prove (a) and (b).
    First, to see (a), let $a\in\matrices$ be fixed. 
    Then for almost every $\omega\in\Omega$, from $\opLdag(u) = \alpha u$, we have
    \begin{align*}
        \tr{\phi_{\theta(\omega)}\!\seq{u_\omega\varrho_\omega }
            a}
        = 
            \tr{u_\omega\varrho_\omega
            \phi_{\theta(\omega)}^\dagger(a)}
            &= 
            \overline{\alpha}
            \tr{
            \varrho_\omega \phi_{\theta(\omega)}^\dagger(a)\phi_{\theta(\omega)}^\dagger(u_{\theta(\omega)})}
            \\
        &= 
        \overline{\alpha}
            \tr{
            \varrho_\omega \phi_{\theta(\omega)}^\dagger(au_{\theta(\omega)})} 
            &&\text{by Lemma \ref{Lem:Eigenmat_of_opLdag_in_multdom}}\\
            &= 
            \overline{\alpha}
         \tr{
            \varrho_{\theta(\omega)} au_{\theta(\omega)}}
            &&\text{since $\opL(\varrho) = \varrho$}\\
            &= 
            \overline{\alpha}
         \tr{
            u_{\theta(\omega)}\varrho_{\theta(\omega)} a}.
    \end{align*}
    Thus, because $a\in\matrices$ was arbitrary, we conclude that $\phi_{\theta(\omega)}\!\seq{u_\omega\varrho_\omega }
    =
    \overline{\alpha}u_{\theta(\omega)}\varrho_{\theta(\omega)}$ holds almost surely, i.e., $\opL(u\varrho) = \overline{\alpha} u\varrho$.
    To see (b), note that since $\opL_1$ is a contraction on a Banach space and $\dual{\opL_1} = \opL_1^\dagger$, the general fact about contractions on Banach spaces \cite[Proposition 3.1]{Groh1984UniformlyC-algebras} implies that $\Lambda_\opL \subseteq \Lambda_{\opL}^\dagger$.
    Therefore, since from Proposition \ref{Prop:Polar_decomp} (a) we know that every eigenvalue of $\opL_1$ is simple, we conclude from part (a) above and the uniqueness in Proposition \ref{Prop:Polar_decomp} (b) that $\opLdag(u) = \overline{\alpha} u$, as claimed. 
\end{proof}
Therefore, by Proposition \ref{Prop:Lambda_opL_dagger_is_a_group}, we conclude from the above that $\Lambda_\opL$ is a group, which we record as a corollary. 
\begin{cor}\label{Cor:Lambda_is_a_group}
    Assume $\opL$ is irreducible. 
    Then $\Lambda_\opL$ forms a group under multiplication.
\end{cor}
With this in hand, we conclude this section by proving basic group-theoretic properties of $\Lambda_\opL$. 
Let $\Gamma_\opL$ be the quotient group 
\begin{equation}
    \Gamma_\opL :=
    \Lambda_\opL / \Lambda_\theta,
\end{equation}
which coincides with $\GammaGroup$ from the introduction. 
We write $\alpha\Lambda_\theta$ to denote the coset in $\Gamma_\opL$ represented by $\alpha\in\Lambda_\opL$. 
Our goal now is to prove that $\Gamma_\opL$ is finite. 
Recall that $G_\opL$ was defined to be the set of unitary eigenmatrices of $\opLdag$. 
Let $H_\theta$ denote the subgroup of $G_\opL$ defined by 
\begin{equation}
    H_\theta 
        =
    \set{f\mbI
        \,\,:\,\,
    f\in L^\infty(\Omega)\text{ satisfies }|f| = 1\text{ and }\koopman{\theta}(f) = \beta f\text{ for some }\beta\in\Lambda_\theta 
    }.
\end{equation}
Then we have the following reformulation of $\Gamma_\opL$. 
For $u\in G_\opL$, write $u H_\theta$ to denote the coset representative of $u$ in $G_\opL/H_\theta$.
\begin{lemma}\label{Lem:G_mod_H_is_isomorphic_to_Gamma}
    Assume $\opL$ is irreducible. 
    The map 
    \begin{equation}
        \begin{split}
            G_\opL/H_\theta &\to \Gamma_\opL\\ 
            u H_\theta &\mapsto \text{$\alpha\Lambda_\theta$ where $\alpha\in\Lambda_\opL$ is such that $\opLdag(u) = \alpha u$}
        \end{split}
    \end{equation}
    is a well-defined group isomorphism. 
\end{lemma}
\begin{proof}
    We first show well-definedness. 
    Let $u, v\in G_\opL$ satisfy $u H_\theta = v H_\theta$. 
    Then there is $f\in H_\theta$ such that $u = fv$. 
    So, if we let $\alpha, \gamma, \beta\in\Lambda_\opL$ satisfy $\opLdag(u) = \alpha u$, $\opLdag(v) = \gamma v$, and $\koopman{\theta}(f) = \beta f$, then $u = fv$ implies $\alpha = \gamma\beta$. 
    Therefore, $\alpha\Lambda_\theta = \gamma\Lambda_\theta$, so the map is well-defined. 
    Next, this map is a group homomorphism, since $\opLdag(u) = \alpha u$ and $\opLdag(v) = \beta v$ for $\alpha, \beta\in\Lambda_\opL$ implies $\opLdag(uv) = \alpha\beta uv$ by
    Lemma \ref{Lem:Eigenmat_of_opLdag_in_multdom}. 
    Thus, the image of $uv H_\theta$ under the above map is $\alpha\beta \Lambda_\theta = (\alpha \Lambda_\theta)(\beta\Lambda_\theta)$. 
    It is immediate this homomorphism is surjective---by the definition of $G_\opL$---and injectivity follows from the observation that if $\opLdag(u) = \alpha u$ for $\alpha \in\Lambda_\theta$ (i.e., $u H_\theta\mapsto \Lambda_\theta$), then Proposition \ref{Prop:Lambda_opL_is_Lambda_opL_dagger} implies that $\opL(u\varrho) = \overline{\alpha} u \varrho$, where $\varrho$ is the unique steady state of $\opL$. 
    However, $\opL(f\varrho) = \overline{\alpha} f \varrho$ for some $f\in L^\infty(\Omega)$ with $|f| = 1$ almost surely, by Lemma \ref{Lem:Eigenspectrum_koopman_in_opL}. 
    Therefore, by Proposition \ref{Prop:Polar_decomp} (a), we conclude that $u H_\theta = f H_\theta = H_\theta$, which gives the desired injectivity. 
\end{proof}
\begin{lemma}\label{Lem:The unitary representation is faithful}
    Assume $\opL$ is irreducible with unique steady state $\varrho$. 
    Then any $u\in G_\opL\setminus H_\theta$ satisfies $\tr{u \varrho} = 0$ almost surely. 
\end{lemma}
\begin{proof}
    Let $u\in G_\opL$ and $\alpha\in\Lambda_\opL$ be such that $\opLdag(u) = \alpha u$. 
    By Proposition \ref{Prop:Lambda_opL_is_Lambda_opL_dagger}, we know that $\opL(u\varrho) = \overline{\alpha}u \varrho$, hence 
    \begin{equation}\label{Eqn:Trace_is_eigenfunction}
        \overline{\alpha}\operatorname{Tr}\seq{u\varrho}
        =
        \tr{\opL(u\varrho)}
        =
        \tr{\koopman{\theta^{-1}}\!\seq{u\varrho}}.
    \end{equation}
    Because $u \not\in H_\theta$, we know that $\alpha\not\in\Lambda_\theta$, hence $\overline{\alpha}\not\in\Lambda_\theta$. 
    Thus, (\ref{Eqn:Trace_is_eigenfunction}) implies $\tr{u\varrho} = 0$ almost surely, as claimed. 
\end{proof}
\begin{prop}\label{Prop:Gamma is a finite group}
    Assume $\opL$ is irreducible. 
    Then $|\Gamma_\mfL|\leq d^2$.
\end{prop}
\begin{proof}
    By the previous lemma, it suffices to show $G_\opL/H_\theta$ is finite and bounded in size by $d^2$. 
    To do this, we show that any subset $S\subset G_\opL/H_\theta$ satisfies $|S|\leq d^2$.
    For any subset $S\subset G_\opL/H_\theta$ and an element $s\in G_\opL/H_\theta$, let $u_s\in G_\opL$ be a choice of representative of $s$ in $G_\opL$. 
    Notice that for any $s\neq t$ in $S$, we have that $u_s^*u_t\not\in H_\theta$. 
    In particular, by Lemma \ref{Lem:The unitary representation is faithful}, we know that $\tr{\varrho_\omega u_{s; \omega}^*u_{t; \omega}} = 0$ for almost every $\omega\in\Omega$. 
    On the other hand, we know that $\varrho>0$ almost surely, so 
    \begin{equation}
        \begin{split}
            \matrices\times\matrices &\to \mbC\\
            (a, b) &\mapsto \tr{\varrho_\omega a^* b}
        \end{split}
    \end{equation}
    defines a nondegenerate inner product on $\matrices$. 
    So, with respect to this inner product, we know that $\set{u_{s; \omega}}_{s\in S}$ forms an orthogonal set, hence $|\set{u_{s; \omega}}_{s\in S}|\leq d^2$ almost surely. 
    Therefore, $|S| = \avg{|\set{u_{s}}_{s\in S}|} \leq \operatorname{dim}(\matrices) = d^2$, which concludes the proof. 
\end{proof}
With Proposition \ref{Prop:Gamma is a finite group} in hand, the following proposition follows from basic group theory. 
We say that $\Lambda_\theta$ is torsion if for all $\beta\in\Lambda_\theta$, there is $n\in\mbN$ such that $\beta^n = 1$. 
Examples of $\theta$ with $\Lambda_\theta$ torsion include all weakly mixing $\theta$, in addition to cyclic permutations of finite sets. 
\begin{prop}\label{Prop:Cyclic group Gamma when torsion}
    If $\Lambda_\theta$ is torsion, then $\Gamma_\opL$ is a cyclic group.
\end{prop}
\begin{proof}
    See Appendix \ref{App:Group}.
\end{proof}

%%%
 
\section{Periodic properties}\label{Sec:Spectral theory}
Now that we have established the basic structural properties of $\Lambda_\opL$, we may now describe the way in which the periodic properties of $\opL$ (hence the corresponding ergodic quantum process $\eqp$) are encoded in it. 
To streamline presentation, let us introduce some notation and terminology.
Note that, as a consequence of Proposition \ref{Prop:Gamma is a finite group}, for all $\alpha\in\Lambda_\opL$, there is a unique minimal $N\in\mbN$ such that $\alpha^N\in\Lambda_\theta$: we write $N_\alpha$ to denote this minimal $N$. 
By a similar token, from Lemma \ref{Lem:G_mod_H_is_isomorphic_to_Gamma}, for all eigenmatrices $u\in G_\opL$ corresponding to $\alpha\in\Lambda_\opL$, $u^{N_\alpha} = f\mbI$ for some eigenfunction $f$ of $\koopman{\theta}$ corresponding to $\alpha^{N_\alpha}$. 
Because we make reference to this collection of data so frequently, we shall give it a name.
\begin{definition}[Eigentuples of \texorpdfstring{$\opL$}{l}]\label{Def:Eigentuples}
    We call a tuple $\seq{\alpha, \beta, u, f}$ an eigentuple for $\opL$ if $\alpha\in\Lambda_\opL$, $\beta = \alpha^{N_\alpha}$, $u\in G_\opL$ is an eigenmatrix corresponding to $\alpha$, and $f$ is the eigenfunction of $\koopman{\theta}$ corresponding to $\beta\in\Lambda_\theta$ satisfying $u^{N_\alpha} = f\mbI$. 
\end{definition}
For $t\in\mbR$, write $\expnon{t}$ to denote $e^{2\pi i t}$, and for $N\in\mbN$, write $\expn{N}{t}$ to denote $e^{2\pi i t/N}$.
If $f:\Omega\to\mbT$ is a measurable function, there is a unique measurable function $a_f:\Omega\to [0, 1)$ such that $f = e(a_f)$.
In particular, if $f$ is an eigenfunction of $\theta$ with eigenvalue $\expnon{b}$, then by ergodicity we know that $|{f}|$ is constant, so we know that $f = |f|\expnon{a_f}$ for some measurable function $a_f:\Omega\to [0, 1)$ with $\expnon{a_f\circ\theta} = \expnon{a_f + b}$. 
Given $\gamma = e(t)\in\mbT$ and $n\in\mbN$, we let 
\begin{equation}
    \operatorname{roots}_n(\gamma)
    =
    \set{\expn{n}{t + k}}_{k\in\mbZ/n\mbZ},
\end{equation}
where we have identified $\mbZ/n\mbZ$ with $\set{0, \dots, n -1}$. 
We call a function $f:\Omega\to\mbC$ unimodular if $|f| = 1$ almost surely. 
If $f$ is a unimodular function, we write $\operatorname{roots}_n(f)$ to denote the random set defined by $\operatorname{roots}_n(f)_\omega = \operatorname{roots}_n(f(\omega))$.
The following key technical result gives us a first step for extracting periodicity from eigenspectrum. 
\begin{prop}\label{Prop:Almost sure spectrum of Ualphas}
    Assume $\opL$ is irreducible, and let $\seq{\alpha, \beta, u, f}$ be an eigentuple for $\opL$. 
    For a.e. $\omega\in\Omega$, 
    \begin{equation}
        \sigma(u_{\omega})
        =
        \operatorname{roots}_{{N_\alpha}}(f(\omega)).
    \end{equation} 
    In particular, $|\sigma(u_{\omega})| = {N_\alpha}$ and $\sigma(u_{\theta(\omega)}) = \alpha \sigma(u_{\omega})$ for a.e. $\omega\in\Omega$. 
\end{prop}
\begin{proof}
    From Lemmas \ref{Lem:Mult_dom_direct_integral} and \ref{Lem:Eigenmat_of_opLdag_in_multdom}, $u_\omega\in\multdom{\phi_{\omega}^\dagger}$ for a.e. $\omega\in\Omega$. 
    We claim that, for a.e. $\omega\in\Omega$, the map $\phi_{\theta(\omega)}^\dagger$ defines an isometric $*$-isomorphism of $C^*$-algebras
    \begin{equation}\label{Eqn:Isometric_phidagger}
        \phi_{\theta(\omega)}^\dagger:C^*(u_{\theta(\omega)})
        \to 
        C^*(u_{\omega}),
    \end{equation}
    where for $a\in\matrices$, $C^*(a)$ denotes the $C^*$-algebra generated by $a$. 
    Indeed, the map (\ref{Eqn:Isometric_phidagger}) is a $*$-homomorphism by the positivity of $\phi_{\theta(\omega)}^\dagger$ and the fact that $u_\omega\in\multdom{\phi_{\omega}^\dagger}$ for a.e. $\omega$, and because $\phi_{\theta(\omega)}^\dagger(u_{\theta(\omega)}) = \alpha u_\omega$ for a.e. $\omega$, we have that (\ref{Eqn:Isometric_phidagger}) is a $*$-isomorphism.
    Because $*$-isomorphisms between $C^*$-algebras are isometric \cite[Theorem 2.1.7]{Murphy2007C-AlgebrasTheory}, so the claim is proved. 
    Therefore, by the claim, for a.e. $\omega\in\Omega$, if $p\in\matrices$ is any spectral projection of $u_{\theta(\omega)}$, $p\in C^*(u_{\theta(\omega)})$ hence $\phi_{\theta(\omega)}^\dagger(p)$ is a nonzero projection in $C^*(u_{\omega})$. 
    So, if we let $\set{p_{\lambda; \omega}}_{\lambda\in\sigma(u_{\omega})}$ denote the set of eigenprojections for $u_{\omega}$, then when we write 
    \begin{equation}
        u_{\omega}
        =
        \sum_{\lambda\in\sigma(u_{\omega})}
        \lambda p_{\lambda; \omega},
    \end{equation}
    the equation $\opLdag(u) = \alpha u$ gives  
    \begin{align}
         \sum_{\lambda\in \sigma(u_{\omega})}
        \alpha \lambda p_{\lambda; \omega}
        = 
        \alpha u_{\omega}
        &= 
        \phi_{\theta(\omega)}^\dagger(u_{\theta(\omega)})
        =
        \sum_{\lambda\in\sigma(u_{\theta(\omega)})}
        \lambda \phi_{\theta(\omega)}^\dagger(p_{\lambda; \theta(\omega)})
    \end{align}
    for almost every $\omega\in\Omega$. 
    So, since $\phi_{\theta(\omega)}^\dagger(p_{\lambda; \theta(\omega)})$ is a nonzero projection, we conclude that 
    \begin{equation}
         \sigma(u_{\theta(\omega)})
        \subseteq 
        \alpha\sigma(u_{\omega})
    \end{equation}
    holds for almost every $\omega\in\Omega$. 
    In particular, $\omega\mapsto g(\omega) = |\sigma(u_{\omega})|$ satisfies $0 \leq \koopman{\theta}(g)\leq g \leq d$, so by Lemma \ref{Lem:Increasing_implies_constant}, $g$ is almost surely constant. 
    Thus, the almost sure inclusion above is actually the almost sure equality
    \begin{equation}
        \sigma(u_{\theta(\omega)})
        = 
        \alpha\sigma(u_{\omega})
    \end{equation}
    To conclude the proof, therefore, all that remains is to show that $\sigma(u_{\omega}) =  \operatorname{roots}_{{N_\alpha}}(f(\omega))$ for almost every $\omega\in\Omega$. 
    To see this, we begin by noting $u^{{N_\alpha}} = f \mbI$ immediately implies the almost sure inclusion 
    \begin{equation}
        \sigma(u_{\omega})
            \subseteq 
        \operatorname{roots}_{{N_\alpha}}(f(\omega)).
    \end{equation}
    So, to show these sets are almost surely equal, it suffices to show that $|\sigma(u)| = {N_\alpha}$ almost surely.
    Now, consider the random polynomial
    \begin{equation}
        \chi_\omega(x)
            := 
        \prod_{\lambda\in\sigma(u_{\omega})}
        (x - \lambda)
        =
        \det(u_\omega)
        +
        \sum_{k=1}^\ell a_{k; \omega}x^k,
    \end{equation}
    where $\ell$ is the almost everywhere value of $|\sigma(u)|$, and $\operatorname{det}$ denotes the determinant.
    Then $\chi_\omega(u_{\omega}) = 0$ almost surely, so, letting $\varrho\in\rstates$ denote the unique fixed state of $\opL$, we know $\tr{u^k\varrho} = 0$ almost surely for all $k = 1, \dots, {N_\alpha} - 1$ by Lemma \ref{Lem:The unitary representation is faithful}. 
    So if, for a contradiction, $|\sigma(u)| = \ell < {N_\alpha}$, we would have that 
    \begin{align}
        0 = \tr{\chi_\omega(u_\omega)\varrho_\omega}
        &= 
        \operatorname{det}(u_\omega) + 
        \sum_{k=1}^\ell 
        a_{k; \omega}\tr{u_\omega^k\varrho_\omega}\\
        &= 
        \operatorname{det}(u_\omega)
    \end{align}
    almost surely, which contradicts the unitarity of $u$.
    Thus, it must be that $|\sigma(u)| = \ell = {N_\alpha}$ almost surely, which concludes the proof. 
\end{proof}
We may now give a proof of Theorem \ref{Thm:Finite_group}.
Recall the statement of this theorem. 
\FiniteGroup*
\begin{proof}
    Note that $\GammaGroup = \Gamma_\opL$, so $|\GammaGroup|\leq d^2$ follows from Proposition \ref{Prop:Gamma is a finite group}.
    The fact that each $\alpha\in\PerSpecEQP$ is simple follows from Proposition \ref{Prop:Polar_decomp} (a), and the fact that $N_\alpha \leq d$ follows from Proposition \ref{Prop:Almost sure spectrum of Ualphas} since $N_\alpha = \avg{|\sigma(u)|}\leq d$. 
\end{proof}
We shall also now take this opportunity to note the following improvement of Propositions \ref{Prop:Gamma is a finite group} and \ref{Prop:Cyclic group Gamma when torsion}.
\begin{cor}\label{Cor:theta_torsion_group_order_leq_d}
    Assume $\opL$ is irreducible. 
    If $\Lambda_\theta$ is torsion, then $\Gamma_\opL$ is a cyclic group with $|\Gamma_\opL|\leq d$.
\end{cor}
\begin{proof}
    By Proposition \ref{Prop:Cyclic group Gamma when torsion}, we know that $\Lambda_\theta$ being torsion implies $\Gamma_\opL$ is a finite cyclic group. 
    By Proposition \ref{Prop:Almost sure spectrum of Ualphas}, we know that ${N_\alpha}\leq d$ for any $\alpha\in\Lambda_\opL$.
    Thus, every element of the cyclic group $\Gamma_\opL$ has order bounded by $d$, which shows $|\Gamma_\opL|\leq d$. 
\end{proof}
% %
% %
% In particular, whenever $\theta$ is weakly mixing, $\Gamma_\opL$ is a finite cyclic group of order at most $d$.
%
%
\begin{remark}\label{Rem:Sharp}
    This upper bound is sharp: indeed, let $\set{\ket{e_i}}_{i=1}^d$ denote the standard basis of $\mbC^d$ and let $\psi:\matrices\to\matrices$ be the map defined by 
    \begin{equation}
        \psi(a)
        =
        \sum_{i=1}^dv_i a v_i^*
    \end{equation}
    where $v_i = \ketbra{e_i}{e_{i+1}}$ for $i = 1, \dots, d-1$ and $v_d = \ketbra{e_d}{e_1}$.
    One immediately verifies that $\psi$ defines a unital quantum channel. 
    Moreover, $\psi$ is irreducible: to see this, it suffices to notice that $d^{-1}\mbI$ is the unique element of $\states$ in $\operatorname{Fix}(\psi)$, which follows from the fact that, if $\psi(\rho) = \rho$ for some $\rho\in\states$, $\rho$ is necessarily diagonal, so by the structure of $\psi$, $\rho = d^{-1}\mbI$. 
    It is clear that if we let $\gamma = \exp({2\pi i/d})$ and define $w\in\matrices$ to be the unitary matrix $w = \sum_{k=1}^d \gamma^k \ketbra{e_k}{e_k}$, then $\psi(w) = \gamma w$, so by \cite{Evans1978SpectralC-Algebras} we have that $\perspec{\psi} \cong \mbZ/d\mbZ$, i.e., $|\Gamma_\opL| = d$. 
    (In this case, $\Omega = \set{1}$ is the trivial probability space so $\opL = \psi\in\bops{\matrices}$, and $\Gamma_\opL = \perspec{\psi}$.)
\end{remark}
We now prove Theorem \ref{Thm:Any_group}. 
\AnyGroup*
\begin{proof}
    For any abelian group $G$, we have that $G\cong \mbZ/n_1\mbZ\oplus\cdots\oplus\mbZ/n_k\mbZ$ for some $\set{n_1, \dots, n_k}\subset\mbN$. 
    Let $\Omega = \mbT^k$ with the uniform measure and define $\theta:\mbT^k\to\mbT^k$ by 
    \begin{equation*}
        \theta(\omega_1, \dots, \omega_k)
        =
        (\omega_1 + \alpha_1, \dots, \omega_k + \alpha_k)
    \end{equation*}
    for some rationally independent $\alpha_1, \dots, \alpha_k\in [0, 1]\setminus\mbQ$, where we have identified $\mbT$ with $\mbR/\mbZ$. 
    It is a standard fact that $\theta$ defines an ergodic map for the uniform measure, and that 
    \begin{equation*}
        \Lambda_\theta 
        =
        \set{
            e^{2\pi i (m_1\alpha_1 + \cdots + m_k\alpha_k)}
            \,\,:\,\, m_1, \dots, m_k\in\mbZ
        }.
    \end{equation*}
    Now let $d = n_1\cdots n_k$, and identify $\matrices$ with the $k$-fold tensor product 
    \begin{equation*}
        \matrices 
        \cong 
        \mbM_{n_1}\otimes\cdots\otimes \mbM_{n_k}.
    \end{equation*}
    For all $j\in\set{1, \dots, k}$, let $\Sigma^{(j)}_1$ and $\Sigma^{(j)}_3$ be the shift and clock matrices on $\mbM_{n_j}$, i.e., 
    \begin{equation*}
        \Sigma^{(j)}_1
        =
        \left[ 
        \begin{array}{cccccc}
            0 & 0 & 0 & \cdots & 0 & 1\\
            1 & 0 & 0 & \cdots & 0 & 0\\
            0 & 1 & 0 & \cdots & 0 & 0\\
            0 & 0 & 1 & \cdots & 0 & 0\\
            \vdots & \vdots & \vdots & \ddots & \vdots& \vdots\\
            0 & 0 & 0 & \cdots & 1 & 0
        \end{array}
        \right]
        \quad 
        \text{and}\quad 
        \Sigma^{(j)}_3
        =
        \operatorname{diag}(1, e^{2\pi i / n_j}, \dots, e^{2\pi i(n_j - 1) / n_j}). 
    \end{equation*}
    These satisfy the commutation relations $\Sigma^{(j)}_3\Sigma^{(j)}_1 = e^{2\pi i/n_j}\Sigma^{(j)}_1\Sigma^{(j)}_3$, and $(\Sigma^{(j)}_1)^{n_j} = (\Sigma^{(j)}_3)^{n_j} = I_{n_j}$. 
    We then define 
    \begin{equation*}
        \begin{split}
            S_j 
            =
        I_{n_1}\otimes\cdots \otimes I_{n_{j-1}}\otimes \Sigma^{(j)}_1\otimes I_{n_{j+1}}\otimes\cdots I_{n_k}, \\
        C_j 
        =
        I_{n_1}\otimes\cdots \otimes I_{n_{j-1}}\otimes \Sigma^{(j)}_3\otimes I_{n_{j+1}}\otimes\cdots I_{n_k}
        \end{split}
    \end{equation*}
    which are in $\mbM_d$. 
    Define 
    \begin{equation*}
        \Delta:\matrices\to\matrices,\quad \Delta(a)
        =
        \frac{1}{d}
        \sum_{j_1=1}^{n_1}
        \cdots
        \sum_{j_k = 1}^{n_k}
        C_{j_1}\cdots C_{j_k}
        a
        (C_{j_1}\cdots C_{j_k})^*,
    \end{equation*}
    which projects $a$ onto its diagonal components.

    For $j\in\set{1, \dots, k}$, define $f_j:\mbT^k\to\mbC$ by 
    \begin{equation*}
        f_j(\omega_1, \dots, \omega_k)
            =
        e^{2\pi i \omega_j/n_j},
    \end{equation*}
    and notice that there is a measurable function $c_j:\mbT^k\to\set{e^{2\pi i m/n_j}\,\,:\,\, m\in\set{1, \dots, n_j}}$ depending only on the $j$th coordinate such that 
    \begin{equation*}
        f_j(\theta(\omega_1, \dots, \omega_k))
        =
         e^{2\pi i (\omega_j + \alpha_j))/n_j}
        =
        e^{2\pi i \alpha_j/n_j}
        c_j(\omega_1, \dots, \omega_k)f_j(\omega_1, \dots, \omega_k).
    \end{equation*}
    Write $c_j(\omega_1, \dots, \omega_k) = e^{2\pi i p_j(\omega_1, \dots, \omega_k)/n_j}$ for a measurable function $p_j:\mbT^k\to\set{1, \dots, n_j}$. 
    We may now define our ergodic quantum process via the random quantum channel 
    \begin{equation*}
        \phi_{\omega}(a)
            := 
        V_\omega \Delta(a)V_\omega^*
        \quad \text{where }\quad 
        V_\omega 
        =
        S_1^{p_1(\omega)}\cdots S_k^{p_k(\omega)}.
    \end{equation*}
    Because $V$ is unitary, this is a unital quantum channel. 
    We now show that the quantum process defined by $\phi$ is irreducible, and since $\phi$ is a unital quantum channel, it suffices to show that 
    \begin{equation*}
        \phi_{\theta(\omega)}(x_{\theta(\omega)})
        =
        x_{\omega}
    \end{equation*}
    implies $x\in\mbC I$, i.e., we show the adjoint process is irreducible. 
    To do this, we first note that $\set{S_{\boldsymbol{v}} C_{\boldsymbol{u}}}_{\boldsymbol{v}, \boldsymbol{u}}$ forms an orthogonal basis of $\matrices$, where for a multi-index $\boldsymbol{v} = (v_1, \dots, v_k)\in [n_1]\times\cdots \times[n_k]$, we take 
    \begin{equation*}
        S_{\boldsymbol{v}}
        =
        \prod_{j=1}^k S_{j}^{v_j}
    \end{equation*}
    and $C_{\boldsymbol{u}}$ is defined similarly. 
    So, any $x$ may be written 
    \begin{equation*}
        x 
            =
        \sum_{\boldsymbol{v}, \boldsymbol{u}}
        g_{\boldsymbol{v}, \boldsymbol{u}} S_{\boldsymbol{v}}  C_{\boldsymbol{u}}
    \end{equation*}
    where $g_{\boldsymbol{v}, \boldsymbol{u}}:\mbT^k\to\mbC$ is measurable. 
    Because $\Delta$ projects onto the diagonal components, we see that 
    \begin{equation*}
        \phi_{\theta(\omega)}(x_{\theta(\omega)})
        =
        \phi_{\theta(\omega)}\seq{
        \sum_{\boldsymbol{u}}
        g_{\boldsymbol{0}, \boldsymbol{u}}(\theta(\omega))C_{\boldsymbol{u}}}.
    \end{equation*}
    Therefore, by the commutation relations between the clock and shift matrices, and the fact that $\phi_{\theta(\omega)}(x_{\theta(\omega)})
        =
        x_{\omega}$, we see that 
        \begin{equation*}
            g_{\boldsymbol{0}, \boldsymbol{u}}(\omega)
            =
             g_{\boldsymbol{0}, \boldsymbol{u}}(\theta(\omega))
             \prod_{j=1}^k c_{j}(\omega)^{-u_j}
        \end{equation*}
    holds almost surely. 
    But by definition, $c_{j}(\omega)^{-u_j} = e^{2\pi i u_j\alpha_j/n_j} f_j(\theta(\omega))^{-u_j} f_j(\omega)^{u_j}$, whence we conclude from the above that 
    \begin{equation*}
        h_{\boldsymbol{u}}(\theta(\omega))
        =
         \seq{\prod_{j=1}^k e^{-2\pi i u_j\alpha_j/n_j} }
         h_{\boldsymbol{u}}(\omega)
    \end{equation*}
    holds almost surely, where $h_{\boldsymbol{u}}(\omega) =  g_{\boldsymbol{0}, \boldsymbol{u}}(\omega)\prod_{j=1}^k f_j(\omega)^{-u_j}$. 
    However, $\prod_{j=1}^k e^{-2\pi i u_j\alpha_j/n_j}\in \Lambda_\theta$ if and only if $u_j = 0$ for all $j$, from which we conclude $h_{\boldsymbol{u}} = 0$ almost surely hence $g_{\boldsymbol{0}, \boldsymbol{u}} = 0$ almost surely whenever $\boldsymbol{u}\neq 0$. 
    This also shows that $g_{0, 0} = \lambda$ almost surely for some $\lambda\in\mbC$, so $x\in\mbC I$ and $\Phi$ is irreducible. 
    To see that $\GammaGroup$ for this $\Phi$ is isomorphic to $G$, we simply note that if for $\boldsymbol{m}\in [n_1]\times\cdots\times[n_k]$ we define 
    \begin{equation*}
        u_{\boldsymbol{m};\omega}
            :=
        \prod_{j=1}^k 
        e^{2\pi i \omega_j m_j/n_j}
        C_j^{m_j},
    \end{equation*}
    then we see that $u_{\boldsymbol{m}}(\omega)$ is invariant under $\Delta$, and therefore by the commutation relations between the clock and shift matrices, one computes 
    \begin{equation}
        \phi_{\theta(\omega)}(u_{\boldsymbol{m}; \theta(\omega)})
        =
        \seq{ \prod_{j=1}^k 
        e^{2\pi i \alpha_j m_j/n_j}}
        u_{\boldsymbol{m}; \omega},
    \end{equation}
    which shows that $\prod_{j=1}^k 
        e^{2\pi i \alpha_j m_j/n_j}\in \Lambda_\Phi$, and it is then clear that $\GammaGroup\cong G$, concluding the proof. 
\end{proof}
\subsection{Projections encoding periodicity}
Let us extract from Proposition \ref{Prop:Almost sure spectrum of Ualphas} an important piece of information, namely, that there is a measurable way to select the spectral projections of $u_\omega$. 
Specifically, for $\alpha\in\Lambda_\opL$, let $\mcP_\alpha = \set{p_k}_{k\in\mbZ/N_\alpha\mbZ}$ be the set of random projections defined as follows: 
let $\seq{\alpha, \beta, u, f}$ be an eigentuple for $\Phi$. 
For $k\in\mbZ/N_\alpha\mbZ$, we have that $\omega\mapsto \expn{N_\alpha}{a_f(\omega) + k}$ is a measurable function, and, moreover, $\expn{N_\alpha}{a_f(\omega) + k}\in\sigma(u_\omega)$ almost surely. 
We define $p_k\in\Linfty{\Omega}{\matrices}$ to be the random projection such that $p_{k; \omega}$ is spectral projection of $u_{\alpha; \omega}$ corresponding to $\expn{N_\alpha}{a_f(\omega) + k}$.
This is independent of choice of $u\in G_\opL$ since $\opLdag(v) = \alpha v$ for some $v\in G_\opL$ implies $v = \gamma u$ for some $\gamma\in\mbT$, hence the spectral projections of $v$ and $u$ coincide almost everywhere. 
Note also that, in the above, we saw that $\phi^\dagger_{\theta(\omega)}(p_{k; \theta(\omega)})$ was a projection for all $k$, due to $\phi^\dagger_{\theta(\omega)}: C^*(u_{\theta(\omega)})\to C^*(u_{\omega})$ being a $C^*$-algebra isomorphism. 
More generally, for any $n$, this shows that $\phi^\dagger_{\theta(\omega)}\circ\cdots\circ\phi^\dagger_{\theta^n(\omega)}(p_{k; \theta^n(\omega)})$ is a projection for all $k$. 
That is, $\seq{\opLdag}^n\!(p_k)\in\Linfty{\Omega}{\matrices}$ defines a projection for all $n$. 
We record this discussion as a lemma. 
\begin{lemma}\label{Lem:opLdag_of_p_is_a_projection}
    Assume $\opL$ is irreducible. 
    For any $\alpha\in\Lambda_\opL$ and $p\in\mcP_\alpha$, $\seq{\opLdag}^n\!(p)$ defines a projection for all $n\in\mbN$. 
\end{lemma}
Let $\seq{\alpha, \beta, u, f}$ be an eigentuple for $\opL$, and write $\alpha = \expnon{t}$ for $t\in [0, 1)$. 
Then we have that $e(\koopman{\theta}(a_f)) = e(a_f + N_\alpha t)$, so if we define $\xi_\alpha:\Omega\to\set{-1, 0}$ by
\begin{equation}\label{Eqn:Defn_of_xi_function}
        \xi_{\alpha}
        :=
        \koopman{\theta}(a_f) 
        -
        a_f
        -
        \fract{N_\alpha t}
\end{equation}
where $\fract{s} = s - \floor{s}$ and $\floor{\cdot}$ is the floor function, then we have that 
\begin{equation}
    \koopman{\theta}(a_f)
    =
    a_f + \operatorname{frac}(N_\alpha t) + \xi_\alpha 
\end{equation}
almost surely. 
Notice that $\xi_\alpha$ is independent of choice of $u$, because multiplication of $u$ by an element of $\mbT$ does not affect $\xi_\alpha$. 
We define $\varsigma_\alpha:\Omega\to \mbZ/N_\alpha\mbZ$ by
\begin{equation}
        \begin{split}
            \varsigma_{\alpha; \omega} 
                &\equiv
            -\floor{N_\alpha t} + \xi_{\alpha}(\omega) \mod{N_\alpha}.
        \end{split}
\end{equation}
Given any set $\set{a_k}_{k\in\mbZ/N_\alpha\mbZ}$ of measurable functions $a_k:\Omega\to S$ where $S$ is some measure space, let $a_{k + \varsigma_\alpha}$ denote the function $a_{k + \varsigma_\alpha}:\Omega\to S$ given by 
\begin{equation}
    \begin{split}
        a_{k + \varsigma_\alpha; \omega}
        &=
        a_{k + \varsigma_{\alpha}(\omega); \omega}, 
    \end{split}
\end{equation}
where the subscript $k + \varsigma_{\alpha}(\omega)$ is understood modulo $N_\alpha$. 
\begin{lemma}\label{Lem:Random_permutations}
    Assume $\opL$ is irreducible, let $\seq{\alpha, \beta, u, f}$ be an eigentuple for $\opL$, write $\alpha = e(t)$ for some $t\in [0, 1)$, and enumerate $\mcP_\alpha = \set{p_k}_{k\in\mbZ/N_\alpha\mbZ}$.
    Then for all $k\in\mbZ/N_\alpha\mbZ$,
    \begin{equation}\label{Eqn:Random_permutations}
        \koopman{\theta}(\expn{N_\alpha}{a_f + k}) = \alpha \expn{N_\alpha}{a_f + k + \varsigma_\alpha},
    \end{equation}
    and $\opLdag(p_k) = p_{k + \varsigma_\alpha}$.
\end{lemma}
\begin{proof}
    Since $\koopman{\theta}\seq{a_f} = a_f + \fract{N_\alpha t} + \xi_\alpha$, we see that 
    \begin{align}
        \koopman{\theta}(\expn{N_\alpha}{a_f + k}) 
        =
        \expn{N_\alpha}{a_f + \fract{N_\alpha t} + \xi_\alpha + k}
        &= 
        \expn{N_\alpha}{a_f + N_\alpha t+ k + \varsigma_\alpha} \\
        &= \alpha \expn{N_\alpha}{a_f+ k + \varsigma_\alpha},
    \end{align}
    which is (\ref{Eqn:Random_permutations}).
    So, because $\opLdag(u) = \alpha u$, we conclude   
    \begin{align}
        \sum_{k\in\mbZ/N_\alpha\mbZ}
        \expn{N_\alpha}{a_f + k}
        p_k
            =
         u 
            = 
        \overline{\alpha}\opLdag(u)
            &= 
         \sum_{k\in\mbZ/N_\alpha\mbZ}
         \bar{\alpha} 
         \koopman{\theta}\seq{
          \expn{N_\alpha}{a_f + k}
         }
       \opLdag(p_k)\\
       &= 
         \sum_{k\in\mbZ/N_\alpha\mbZ}
          \expn{N_\alpha}{a_f + k + \varsigma_\alpha}
       \opLdag(p_k).
    \end{align}
    Because $\opLdag(p_k)$ is a projection, we conclude from the above that $\opLdag(p_k) = p_{k + \varsigma_\alpha}$, as claimed. 
\end{proof}
\begin{cor}\label{Cor:varrho_k_results}
    Assume $\opL$ is irreducible with unique steady state $\varrho$, and let $\alpha\in\Lambda_\opL$. 
    Then 
    \begin{equation}
        \phi_{\theta(\omega)}\seq{p_{k; \omega}\varrho_\omega p_{k; \omega}}
    =
    p_{k - \varsigma_{\alpha}(\omega); \theta(\omega)}\varrho_{\theta(\omega)}p_{k - \varsigma_{\alpha}(\omega); \theta(\omega)}
    \end{equation}
    for a.e. $\omega\in\Omega$ and all $p_k\in\mcP_\alpha$.
\end{cor}
\begin{proof}
    From Lemmas \ref{Lem:Mult_dom_direct_integral} and \ref{Lem:Eigenmat_of_opLdag_in_multdom}, we know that $u_\omega\in\mcM_{\phi_\omega^\dagger}$ for a.e. $\omega$, so since $\mcM_{\phi_\omega^\dagger}$ is a $C^*$-algebra, $C^*(u_\omega)\subseteq \mcM_{\phi_\omega^\dagger}$, from which we may conclude that $p_{k; \omega}\in \mcM_{\phi_\omega^\dagger}$ for all $k$ for a.e. $\omega$. 
    Moreover, $\phi_{\theta(\omega)}^\dagger\seq{p_{k - \varsigma_{\alpha}(\omega); \theta(\omega)}}
    =
    p_{k; \omega}$ by Lemma \ref{Lem:Random_permutations}. 
    Thus, for any $a\in\matrices$, 
    \begin{align}
        \innerHS{a}{\phi_{\theta(\omega)}\seq{p_{k; \omega}\varrho_\omega p_{k; \omega}}}
        &=
        \innerHS{\phi_{\theta(\omega)}^\dagger(a)}
        {p_{k; \omega}\varrho_\omega p_{k; \omega}}\\
        &=
        \innerHS{\phi_{\theta(\omega)}^\dagger\seq{p_{k - \varsigma_{\alpha}(\omega); \theta(\omega)
        }ap_{k - \varsigma_{\alpha}(\omega); \theta(\omega)
        }}}
        {\varrho_{\omega}}\\
        &=
        \innerHS{p_{k - \varsigma_{\alpha}(\omega); \theta(\omega)
        }ap_{k - \varsigma_{\alpha}(\omega); \theta(\omega)
        }}
        {\varrho_{\theta(\omega)}}\\
        &=
        \innerHS{a}
        {p_{k - \varsigma_{\alpha}(\omega); \theta(\omega)
        }\varrho_{\theta(\omega)}p_{k - \varsigma_{\alpha}(\omega); \theta(\omega)
        }}.
    \end{align}
    Since $a\in\matrices$ was arbitrary, the desired equality follows. 
\end{proof}
Now, for $\alpha\in\Lambda_\opL$ and $k\in\mbZ/N_\alpha\mbZ$, define the random state $\varrho_{\alpha, k}$ by 
\begin{equation}
    \varrho_{\alpha, k}
    =
    \cfrac{p_k\varrho p_k}{\tr{p_k\varrho}}
\end{equation}
for $p_k\in \mcP_\alpha$.
Note we have used that $\varrho>0$ and $\tr{p_k}\geq 1$ in this definition to get that $\tr{p_k\varrho} > 0$ almost surely. 
\begin{cor}\label{Cor:varrho_k_is_fixed}
    Assume $\opL$ is irreducible.
    For any $\alpha\in\Lambda_\opL$, $\phi_{\theta(\omega)}\seq{\varrho_{\alpha, k; \omega}} = \varrho_{\alpha, k - \varsigma_{\alpha}(\omega); \theta(\omega)}$ for a.e. $\omega\in\Omega$. 
\end{cor}
\begin{proof}
    This follows from Corollary \ref{Cor:varrho_k_results} upon noting that $\phi$ is trace-preserving almost surely. 
\end{proof}
\begin{cor}\label{Cor:varrho_is_sum_of_varrho_k}
    Assume $\opL$ is irreducible with unique steady state $\varrho$. 
    Then 
    \begin{equation}
        \varrho 
        =
        \cfrac{1}{N_\alpha}
        \sum_{k\in\mbZ/N_\alpha\mbZ} 
        \varrho_{\alpha, k}
    \end{equation}
    for all $\alpha\in\Lambda_\opL$. 
\end{cor}
\begin{proof}
    By Corollary \ref{Cor:varrho_k_is_fixed}, we know that $\opL\seq{ \sum_{k\in\mbZ/N_\alpha\mbZ} 
        \varrho_{\alpha, k}} =  \sum_{k\in\mbZ/N_\alpha\mbZ} 
        \varrho_{\alpha, k}.$
    Therefore, by the almost sure equality
    \begin{equation}
        \tr{ \sum_{k\in\mbZ/N_\alpha\mbZ} 
        \varrho_{\alpha, k}} = N_\alpha,
    \end{equation}
    the result follows by the uniqueness of $\varrho$. 
\end{proof}
\begin{cor}\label{Cor:trace_against_p_is_N_alpha_inverse}
    Assume $\opL$ is irreducible with unique steady state $\varrho$, and let $\alpha\in\Lambda_\opL$. 
    Then for all $p\in\mcP_\alpha$, $\tr{\varrho p} = N_\alpha^{-1}$ almost surely.
\end{cor}
\begin{proof}
    By Corollary \ref{Cor:varrho_k_results}, we have that 
    \begin{equation}
        \opL\seq{\sum_{k\in\mbZ/N_\alpha\mbZ}p_k\varrho p_k} = \sum_{k\in\mbZ/N_\alpha\mbZ}p_k\varrho p_k,
    \end{equation}
    so since $\tr{\sum_{k\in\mbZ/N_\alpha\mbZ}p_k\varrho p_k} = 1$ almost surely, we conclude from the uniqueness property of $\varrho$ that $\sum_{k\in\mbZ/N_\alpha\mbZ}p_k\varrho p_k = \varrho$. 
    On the other hand, by Corollary \ref{Cor:varrho_is_sum_of_varrho_k}, we also know that $\varrho = N_\alpha^{-1}\sum_{k\in\mbZ/N_\alpha\mbZ}\varrho_{\alpha, k}$.
    Thus, 
    \begin{equation}
        \frac{1}{N_\alpha}\sum_{k\in\mbZ/N_\alpha\mbZ}\varrho_k
        =
        \varrho
        =
        \sum_{k\in\mbZ/N_\alpha\mbZ}p_k\varrho p_k.
    \end{equation}
    But $\set{p_k}_{k\in\mbZ/N_\alpha\mbZ}$ is a partition of unity, so for all $k$, we conclude that $p_k\varrho p_k = N_\alpha^{-1}\varrho_k$. 
    Taking trace, we find that $\tr{p_k\varrho} = N_\alpha^{-1}$ almost surely, which concludes the proof. 
\end{proof}
This proved, we may now give a proof of Theorem \ref{Thm:Partition}.
\Partition*
\begin{proof}
    The random partition of unity is $\mcP_\alpha$ and the measurable map $\varsigma$ is $\varsigma_\alpha$.
    To see that (\ref{Eqn:Partition_1}) holds, it suffices to show that, for a.e. $\omega$, any projection $q\in p_\omega\matrices p_\omega$ satisfies
    \begin{equation}
        \phi_{\theta(\omega)}(q)\in p_{k - \varsigma_{\alpha}(\omega); \theta(\omega)}\matrices p_{k - \varsigma_{\alpha}(\omega); \theta(\omega)},
    \end{equation}
    because $p_\omega\matrices p_\omega$ is a von Neumann algebra hence is the closed span of its projections.
    So, let $q\in p_\omega \matrices p_\omega$ be a projection.
    Because $\varrho_\omega>0$, there is a constant $C>0$ such that $q \leq C p_\omega\varrho_\omega p_\omega$. 
    Thus,
    \begin{equation}\label{Eqn:Partition_proof_1}
        \phi_{\theta(\omega)}(q)\leq C \phi_{\theta(\omega)}(p_{k; \omega}\varrho_\omega p_{k; \omega})
    =
    C
    p_{k - \varsigma_{\alpha}(\omega); \theta(\omega)}
        \varrho_{\theta(\omega)}
    p_{k - \varsigma_{\alpha}(\omega); \theta(\omega)}
    \end{equation}
    by Corollary \ref{Cor:varrho_k_results}.
    Since $\varrho_{\theta(\omega)}>0$, we conclude that there is a constant $D > 0$ such that 
    \begin{equation}
         p_{k - \varsigma_{\alpha}(\omega); \theta(\omega)}
        \varrho_{\theta(\omega)}
    p_{k - \varsigma_{\alpha}(\omega); \theta(\omega)}
    \leq 
    D
     p_{k - \varsigma_{\alpha}(\omega); \theta(\omega)},
    \end{equation}
    so (\ref{Eqn:Partition_proof_1}) gives $  \phi_{\theta(\omega)}(q)
    \leq 
    CD  p_{k - \varsigma_{\alpha}(\omega); \theta(\omega)}$.
    Since $\phi_{\theta(\omega)}(q)\geq 0$ and $C, D>0$, this implies that the range projection of $\phi_{\theta(\omega)}(q)$ is dominated by $p_{k - \varsigma_{\alpha}(\omega); \theta(\omega)}$, which shows that
    \begin{equation}
    \phi_{\theta(\omega)}(q)
    \in 
    p_{k - \varsigma_{\alpha}(\omega); \theta(\omega)}\matrices 
    p_{k - \varsigma_{\alpha}(\omega); \theta(\omega)},
    \end{equation}
    as desired. 
    Because we have already proved (\ref{Eqn:Partition_2}) in Corollary \ref{Cor:varrho_is_sum_of_varrho_k}, the proof is concluded. 
\end{proof}
\subsection{Regularity of periodicity}
Having established the basic structural properties of the projections $\mcP_\alpha$, we turn to studying the regularity of their periodicity. 
Notice that, by Lemma \ref{Lem:Random_permutations}, if we let 
\begin{equation}
    \omega\mapsto \varsigma_{\alpha, n}(\omega) := \sum_{\ell=0}^{n -1}\varsigma_{\alpha}(\theta^\ell(\omega))
    \mod N_\alpha,
\end{equation}
we see $\seq{\opLdag}^n\!(p_k) = p_{k + \varsigma_{n,\alpha}}$. 
Therefore, if we define $\Omega_\alpha = \Omega \times \mbZ/{N_\alpha\mbZ}$ to be the measure space with $\sigma$-algebra $\mcF_\alpha = \mcF\otimes\Sigma_\alpha$ where $\Sigma_\alpha$ is the discrete $\sigma$-algebra on $\mbZ/N_\alpha\mbZ$, and we give $\Omega_\alpha$ the measure $\mu_\alpha = \mu\otimes\pi_\alpha$ where $\pi_\alpha$ is the uniform measure on $\mbZ/N_\alpha\mbZ$, then, letting $T_\alpha:\Omega_\alpha\to\Omega_\alpha$ be the skew-product map 
\begin{equation}
    T_\alpha(\omega, x)
    =
    \seq{
    \theta(\omega), 
    x - \varsigma_{\alpha}(\omega)
    },
\end{equation}
it is then straightforward to check that $T_\alpha$ preserves $\mu_\alpha$ and $T_\alpha^n(\omega, x) = (\theta^n(\omega), x - \varsigma_{\alpha, n}(\omega))$.
In fact, $T_\alpha$ is actually ergodic for $\mu_\alpha$:
\begin{prop}\label{Prop:Talpha_is_ergodic}
    Assume $\opL$ is irreducible, and let $\alpha\in\Lambda_\opL$. 
    Then $T_\alpha$ is a spectrally maximal ergodic extension of $\theta$ by $N_\alpha$. 
\end{prop}
\begin{proof}
    Let $f:\Omega_\alpha\to\mbC$ be measurable with $f\circ T_\alpha = f$. 
    Define $x\in\Linfty{\Omega}{\matrices}$ by 
    \begin{equation}
        x_\omega 
        =
        \sum_{k\in\mbZ/N_\alpha\mbZ}
        f(\omega, k)p_{k; \omega}.
    \end{equation}
    Then by Lemma \ref{Lem:Random_permutations}, we have
    \begin{align}
        \opLdag(x)_\omega 
        =
         \sum_{k\in\mbZ/N_\alpha\mbZ}
        f(\theta(\omega), k)\opLdag(p_k)_\omega
        &=
         \sum_{k\in\mbZ/N_\alpha\mbZ}
        f(\theta(\omega), k)p_{k + \varsigma_\alpha(\omega); \theta(\omega)}\\
        &=
         \sum_{k\in\mbZ/N_\alpha\mbZ}
        f(T_\alpha(\omega, k))p_{k; \omega}\\
        &=
         \sum_{k\in\mbZ/N_\alpha\mbZ}
        f(\omega, k)p_{k; \omega}\\
        &= x_\omega.
    \end{align}
    Therefore, by Proposition \ref{Prop:Fix_opLdag_is_CI}, there is $\lambda\in\mbC$ such that $x = \lambda \mbI$, from which we conclude $f$ is constant.
    Thus, by Proposition \ref{Prop:Basics_of_ergodicity}, $T_\alpha$ is ergodic. 
    To see that $T_\alpha$ is spectrally maximal, it suffices to notice that by Lemma \ref{Lem:Random_permutations}, we have that $|\Lambda_{T_\alpha}/\Lambda_\theta|\geq N_\alpha$. 
    It is a standard fact that $|\Lambda_{T_\alpha}/\Lambda_\theta|\leq N_\alpha$ (see \cite[Chapter 10]{Cornfeld1982ErgodicTheory}), so the proof is done. 
\end{proof}
We now collect a corollary of this fact. 
\begin{cor}
    Assume $\opL$ is irreducible, and let $\alpha\in\Lambda_\opL$. 
    For all $k\in\mbZ/N_\alpha\mbZ$, the limit 
    \begin{equation}
    \lim_{N\to\infty}
    \cfrac
    {\#
    \set{
    n\in\set{1, \dots, N}\,\,:\,\,
    \varsigma_{\alpha, n}(\omega)
    =
    k
    }
    }
    {N}
    \end{equation}
    exists and is equal to $N_\alpha^{-1}$ for a.e. $\omega\in\Omega$. 
\end{cor}
\begin{proof}
    Clearly, $\#
    \set{
    n\in\set{1, \dots, N}\,\,:\,\,
    \varsigma_{\alpha, n}(\omega)
    =
    k
    }
    =
    \sum_{n=1}^N 
    1_{\Omega\times\set{0}}(T_\alpha^n(\omega, k))$ for all $N\in\mbN$. 
    The result then follows by the pointwise ergodic theorem. 
\end{proof}
As a consequence of the above discussion, if for all $n\in\mbN$ and $\alpha\in\Lambda_\opL$ we define $\tau_{\alpha, n}:\Omega\to\mbN\cup\set{\infty}$ by
\begin{equation}
    \begin{split}
        \tau_{\alpha, n}(\omega)
            &:= 
        \inf\set{k > \tau_{\alpha, n-1}(\omega)\,\,:\,\, \varsigma_{\alpha, k}(\omega) = 0}
    \end{split}
\end{equation}
where we take $\tau_{\alpha, 0} = 0$ and $\inf\emptyset = \infty$ by convention, then we have that $\tau_{\alpha, n}<\infty$ almost surely for all $n$, and 
\begin{equation}
    \lim_{N\to\infty}\cfrac
    {
    \#\seq{
    \set{1, \dots, N}\cap\set{\tau_{\alpha, n}(\omega)}_{n=1}^\infty
    }
    }
    {N}
    = N_\alpha^{-1}
\end{equation}
almost surely. 
Also, by Lemma \ref{Lem:Random_permutations}, we may rewrite $\tau_{\alpha, n}$ as 
\begin{equation}
    \tau_{n, \alpha}(\omega)
    =
    \inf\set{
    k > \tau_{\alpha, n-1}(\omega)\,\,:\,\, \Phi^{(k)\dagger}_\omega\seq{p_{\theta^k(\omega)}}
    =
    p_{\omega}\text{ for all }p\in\mcP_{\alpha}
    },
\end{equation}
so, in this way, we see that $\tau_{\alpha, n}$ is a $\eqp$-stopping time for all $n$. 
\begin{lemma}\label{Lem:taus_are_stopping_times}
    Assume $\opL$ is irreducible, and let $\alpha\in\Lambda_\opL$. 
    For all $n\in\mbN$, $\tau_{n, \alpha}$ is a $\Phi$-stopping time. 
    Moreover, 
    \begin{equation}
        \lim_{N\to\infty}\cfrac
    {
    \#\seq{
    \set{1, \dots, N}\cap\set{\tau_{\alpha, n}(\omega)}_{n=1}^\infty
    }
    }
    {N}
    = N_\alpha^{-1}
    \end{equation}
    holds for $\mu$-almost every $\omega\in\Omega$. 

\end{lemma}
%
%
% Let us also note the following fact about $\tau_{\alpha, n}$.
% %
% %
% \begin{prop}
%     Assume $\opL$ is irreducible, and let $\alpha\in\Lambda_\opL$. 
%     %
%     For all $n\in\mbN$, $\tau_{n, \alpha}\in L^1(\Omega)$.
% \end{prop}
% \begin{proof}
%     {\color{red}\textbf{Under construction}}
%     {\color{red}\textbf{Under construction}}
%     {\color{red}\textbf{Under construction}}
%     {\color{red}\textbf{Under construction}}
% \end{proof}
%
%
We now have all the pieces required to prove Theorem \ref{Thm:Periodicity}. 
Recall its statement. 
\Periodicity*
\begin{proof}
    The $\eqp$-stopping time is $\tau = \tau_{\alpha, 1}$, which by Lemma \ref{Lem:taus_are_stopping_times} has density $N_\alpha^{-1}$. 
    The fact that (\ref{Eqn:Periodicity}) holds follows because $\tau_{\alpha, 1}$ was defined to be the first time so that 
    \begin{equation}
        \varsigma_{\alpha, \tau_{\alpha, 1}} \equiv 0 \mod N_\alpha, 
    \end{equation}
    so Theorem \ref{Thm:Partition} implies (\ref{Eqn:Periodicity}). 
    The equation (\ref{Eqn:Periodicity}) precisely says that the disordered quantum process defined by $\seq{\theta^\tau, \phi^{(\tau)}}$ is reduced by all $p\in\mcP_\alpha$, so in particular, whenever $|\GammaGroup| > 1$, because $|\mcP_\alpha| = N_\alpha$, there is some $\alpha\in\Lambda_\opL$ such that $|\mcP_\alpha|>1$, so because $\tr{p}\geq 1$ almost surely for all $p\in\mcP_\alpha$, it follows that $p\neq\mbI$, hence $\seq{\theta^\tau, \phi^{(\tau)}}$ is reducible (i.e., reduced by a nonidentity projection), which concludes the proof. 
\end{proof}
The following example, which is just a particular example of the class of examples constructed in the proof of Theorem \ref{Thm:Any_group}, shows that $\tau$ is in general nonconstant.
\begin{example}\label{Example:quasiperiodic} 
    Let $\Omega = \mbT$ with the Lebesgue measure, and fix $t\in[0, 1)$ irrational. 
    Let $\theta:\Omega\to\Omega$ be defined by $\theta(\expnon{s}) = \expnon{s + t}$, which, by the irrationality of $t$, is an invertible ergodic mpt. 
    Define $\phi:\Omega\to\linears{\mbM_2}$ by its adjoint via
    \begin{equation}
    \phi_{\expnon{s}}^\dagger(a)
    =
    1_{s\in[0, t)}
    \bigg(
    \ketbra{0}{0}a\ketbra{0}{0}
    +
    \ketbra{1}{1}a\ketbra{1}{1}
    \bigg)
    +
    1_{s\in[t, 1)}
    \bigg(
    \ketbra{0}{1}a\ketbra{1}{0}
    +
    \ketbra{1}{0}a\ketbra{0}{1}
    \bigg)
    \end{equation}
    where $\set{\ket{0}, \ket{1}}$ is any orthonormal basis of $\mbC^2$. 
    Then notice that $\phi^\dagger$ is a unital completely positive map, thus $(\theta, \phi)$ defines an ergodic quantum process. 
    In fact, the ergodic quantum process defined by $\seq{\theta, \phi}$ is irreducible, because it is easy to verify that $\opLdag(x) = x$ for some $x\in\mbM_2(\Omega)$ implies that $x \in\mbC\mbI$. 
    Now let $u:\Omega\to\mbM_2$ be defined by $u_{\expnon{s}} = \expn{2}{s}\ketbra{0}{0}
        -
        \expn{2}{s}\ketbra{1}{1},$ and notice that 
    \begin{equation}
        u_{\theta(\expnon{s})}
        =
        \begin{cases}
           \expn{2}{t} u_{\expnon{s}} &\text{if }s\in [0, 1 - t)\\
           -\expn{2}{t} u_{\expnon{s}}
        &\text{if } s\in [1 - t, 1).
        \end{cases}
    \end{equation}
    Then it is easy to check that $\opLdag(u) = -\expn{2}{t} u$ by construction,
    and that $\mcP_{-e_2(t)} = \set{\ketbra{0}{0}, \ketbra{1}{1}}$.
    But 
    \begin{equation}
        \phi_{\theta(\expnon{s})}^\dagger(\ketbra{0}{0})
        =
        1_{s\in [0, t)}\ketbra{0}{0}
        +
        1_{s\in[t, 1)}\ketbra{1}{1}
    \end{equation}
    hence $\tau_{-e_2(t)} = 1$ on the event $\set{\expnon{s}\,\,:\,\, s\in [0, t)}$---which occurs with probability $t \in (0, 1)$---and $\tau_{-e_2(t)} > 1$ otherwise.
\end{example}
We now prove Theorem \ref{Thm:PeriodicChar}.
\PeriodicChar*
\begin{proof}
    From Proposition \ref{Prop:Talpha_is_ergodic}, $T_\alpha$ is a spectrally maximal ergodic extension of $\theta$ by $N_\alpha$, and Theorem \ref{Thm:Periodicity} yields that the given $p$ reduces $(\theta_\varsigma, \phi_\varsigma)$. 
    To see that $(\omega, k)\mapsto p_{k; \omega}$ is minimal, assume $q(\omega, k)\leq p_{k; \omega}$ almost surely. 
    Then we see that 
    \begin{equation*}
        \phi_{\theta(\omega)}\seq{\sum_{k} q(\theta(\omega), k)}
        =
        \sum_{k}q(\omega, k + \theta_\varsigma(\omega))
        =
        \sum_{k}q(\omega, k)
    \end{equation*}
    holds almost surely, hence $\sum_{k} q(\omega, k)\in\mbC\mbI$ almost surely. 
    Because each $\set{p_{k; \omega}}_{k}$ is a set of orthogonal projections and $q \leq p$, it holds that $q(\omega, k) = 0$ or $q(\omega, k) = p_{k; \omega}$ almost surely, which shows the desired minimality. 
    Lastly, to see the other direction of the theorem, if $\theta_\zeta$ is a maximally spectral ergodic extension of $\theta$ by $N$ for some $N$, then let $\lambda\in\Lambda_{\theta_\zeta}$ be such that $\lambda^N\in\Lambda_{\theta}$ and $\lambda^k\not\in\Lambda_\theta$ for all $k < N$, and let $g:\Omega\times\mbZ/N\mbZ\to\mbC$ be a unimodular eigenfunction corresponding to $\lambda$, i.e., $g\circ\theta_\zeta = \lambda g$ almost surely. 
    Let also $p\in\mathbb{M}_d(\Omega\times\mbZ/N\mbZ)$ be a minimal reducing projection for $(\theta_\zeta, \phi_\zeta)^\dagger$, so we have that $\phi_{\zeta; \theta_\zeta(\omega, k)}^\dagger(p_{\theta_\zeta(\omega, k)}) = p_{\omega, k}$ almost surely via \cite{Ekblad2024ReducibilityProcesses}. 
    Then if we define 
    \begin{equation*}
        u_\omega 
            :=
        \sum_{k\in\mbZ/N\mbZ}g(\omega, k)p(\omega, k),
    \end{equation*}
    we see that $u$ is almost surely unitary and 
    \begin{align*}
        \phi_{\theta(\omega)}(u_{\theta(\omega)})
            &=
        \sum_{k\in\mbZ/N\mbZ}g(\theta_\zeta(\omega, k + \zeta(\omega))) 
         \phi_{\theta(\omega)}(p(\theta(\omega), k))\\
            &= 
        \sum_{k\in\mbZ/N\mbZ}g(\theta_\zeta(\omega, k + \zeta(\omega))) p(\omega, k + \zeta(\omega))\\
         &= 
        \lambda \sum_{k\in\mbZ/N\mbZ}g(\omega, k) p(\omega, k),
    \end{align*}
    hence $\lambda\in\Lambda_\phi$. 
    But $\lambda\not\in\Lambda_\theta$, hence $|\GammaGroup|>1$. 
\end{proof}
\subsubsection{The i.i.d. case}
The theory may be refined further in the case that $\theta$ is an i.i.d. shift, which we now describe. 
Recall how to realize an i.i.d. sequence within our framework: let $\seq{\Xi, \nu, \mcG}$ be a probability space and suppose $\psi:\Xi\to\channels$ is a random quantum channel. 
We define $\Omega = \prod_{k\in\mbZ}\Xi$ with the product $\sigma$-algebra $\mcF = \bigotimes_{k\in\mbZ}\mcG$, and we let $\mu = \bigotimes_{k\in\mbZ}$ be the product measure. 
Then define $\theta:\Omega\to\Omega$ by $\seq{\xi_k}_{k\in\mbZ}\mapsto \seq{\xi_{k+1}}$, and let $\phi:\Omega\to\channels$ be the random quantum channel $\phi_{\seq{\xi_{k}}} = \psi_{\xi_0}$. 
Then it is a standard fact that $\theta$ is weakly mixing (hence ergodic) for $\mu$, hence $\seq{\theta, \phi}$ defines an ergodic quantum process $\eqp$. 
In \cite{Ekblad2024ReducibilityProcesses}, it was shown that, under the assumption of irreducibility, minimal reducing projections for such $\eqp$ are nonrandom, i.e., if $p\in\rmatrices$ is a reducing projection for $\eqp$ and $p$ is minimal with respect to $\leq$, then $p$ is deterministic. 
With this fact in hand, we can prove the following. 
\begin{prop}\label{Prop:IID}
    Assume $\opL$ is irreducible and arises from an i.i.d. ergodic quantum process. 
    Let $\seq{\alpha, \beta, u, f}$ be an eigentuple of $\opL$.  
    Then $u$ is deterministic, i.e., $u = \avg{u}$ almost surely.  
    Moreover, all $p_k\in\mcP_\alpha$ are nonrandom, and 
    \begin{equation}
        \zeta u = \sum_{k\in\mbZ/N_\alpha\mbZ} \expn{N_\alpha}{k} p_k
    \end{equation}
    for some $\zeta\in\mbT$.
\end{prop}
\begin{proof}
    Note that since $\seq{\theta, \phi}$ gives an i.i.d. ergodic quantum process, so does $\seq{\theta^n, \phi^{(n)}}$. 
    Therefore, because all $p_k\in\mbP_\alpha$ are minimal reducing projections of $\seq{\theta^n, \phi^{(N_\alpha)}}$ by Theorem \ref{Thm:PeriodicChar}, by the discussion preceding the statement of the proposition, we know that each $p_k$ is deterministic, i.e., $p_k =\avg{p_k}$. 
    On the other hand, because $\theta$ is weakly mixing, we know that $f = \avg{f}\in\mbT$. 
    Write $f = \expnon{t}\in\mbT$ for some $t\in [0, 1)$, so $a_f = t$. 
    Then $\expn{N_\alpha}{a_f + k} = \expn{N_\alpha}{t} \expn{N_\alpha}{k} = \overline{\zeta} \expn{N_\alpha}{k}$, where $\overline{\zeta} = \expn{N_\alpha}{t}$. 
    By Proposition \ref{Prop:Almost sure spectrum of Ualphas}, we conclude that
     \begin{equation}
        \zeta u = \sum_{k\in\mbZ/N_\alpha\mbZ} \expn{N_\alpha}{k} p_k
    \end{equation}
    as claimed. 
\end{proof}

%%%

%
\section{Final remarks}\label{Sec:Final remarks}
We conclude with some final remarks. 
To aid discussion, let us recall the following version of Theorem \ref{Thm:EHK} that is often used in the literature, which has been proved by many authors in many contexts: see \cite{Wolf2012QuantumTour} for a standard source. 
\begin{thmx}\label{Thmx:Aperiodicity}
    Let $\psi\in\channels$ be an irreducible quantum channel with Perron-Frobenius eigenmatrix $\varrho\in\states$. 
    The following are equivalent. 
    \begin{enumerate}[label = (\alph*)]
        \item  $\psi^n$ is irreducible for all $n\in\mbN$. 
        
        \item $\perspec{\psi} = \set{1}$. 

        \item $\lim_{n}\psi^n(a) = \tr{a}\varrho$ for all $a\in\matrices$.
    \end{enumerate}
\end{thmx}
\begin{proof}
    As mentioned above, this theorem is proved in \cite{Wolf2012QuantumTour}, but we sketch a proof for completeness. 
    The equivalence of (a) and (b) is given by Theorem \ref{Thm:EHK}.
    To see that (c) implies (b), assume that (b) fails.
    Then there is a unitary matrix $u\in\matrices$ and $\alpha\in\mbT\setminus\set{1}$ such that $\psi(u\varrho) = \alpha u\varrho$, hence $\psi^n(u\varrho) = \alpha^n u\varrho$ does not converge, as $\alpha\neq 1$.
    To see that (b) implies (c), note that $\lambda = 1$ is a simple eigenvalue of $\psi$ and all other spectrum of $\psi$ is contained in the set $\set{\lambda\in\mbC\,\,:\,\, |\lambda| < 1}$. 
    In particular, expressing $\psi$ as a matrix in Jordan normal form, we see that $\psi^n$ converges onto the projection onto $V_1$, the eigenspace corresponding to $\lambda = 1$. 
    But this eigenspace is precisely $\mbC\mbI$, and therefore the projection operator in question is the map $a\mapsto \tr{a}\varrho$, which shows (c), concluding the proof. 
\end{proof}
The condition (a) is often called \textit{aperiodicity}, the condition (b) is often called \textit{primitivity}, and condition (c) is sometimes called \textit{strong irreducibility}, although, given their equivalence, each of these terms may be used to refer to the other in the literature.
For our purposes, however, we shall discuss these individual conditions as though they were distinct, and we stick with this terminology. 
\subsection{Aperiodicity and minimality of \texorpdfstring{$\mcP_\alpha$}{l}}
In Theorem \ref{Thm:PeriodicChar}, we have seen that aperiodicity and primitivity are equivalent for ergodic quantum processes. 
We now briefly discuss some ergodic-theoretic facts in this vein. 
Let us write $\tau_\alpha$ to denote $\tau_{1, \alpha}$, and note that since $\tau_\alpha$ is a stopping time, the map $\theta^{\tau_\alpha}:\Omega\to\Omega$ is a well-defined measurable map. 
Then we see that $\tau_{n, \alpha}$ may be obtained recursively by $\tau_{n, \alpha} = \tau_{\alpha}\circ\theta^{\tau_{n-1, \alpha}}$.  
Recall that $\phi^{(\tau_\alpha)}:\Omega\to\channels$ is the random quantum channel
\begin{equation}
    \begin{split}
        \phi^{(\tau_\alpha)}_\alpha:\Omega&\to\channels\\
        \phi^{(\tau_\alpha)}_{\omega}
            &=
        \phi_{\tau_\alpha(\omega); \omega}\circ\cdots\circ \phi_{1; \omega}, 
    \end{split}
\end{equation}
which is measurable by Lemma \ref{Lem:taus_are_stopping_times}. 
We then may consider $\opLdagtau := \glop{\theta^{\tau_\alpha}}{\Phi^{(\tau_\alpha)\dagger}}\in\linears{\rmatrices}$, i.e., 
\begin{equation}
    \opLdagtau\seq{x}_\omega 
    =
    \Phi^{(\tau_\alpha)\dagger}_{\omega}\seq{x_{\theta^{\tau_\alpha}(\omega)}}.
\end{equation}
By the definition of $\tau_\alpha$, we see that all $p\in\mcP_\alpha$ satisfy $\opLdagtau(p) = p$, so by Proposition \ref{Prop:Characterization_of_reducing_proj} we know that $p$ reduces $\opLdagtau$. 
Note also that since $\Phi^{(\tau_\alpha)}$ is a quantum channel, we have that $\banachnorm{\Phi^{(\tau_\alpha)\dagger}}{\infty} = 1$ almost surely by the Russo-Dye theorem \cite[Theorem 3.39]{Watrous2018TheInformation}.
% % %
%
\begin{prop}\label{Prop:Defn_of_mfEtau}
    Assume $\opL$ is irreducible and let $\alpha\in\Lambda_\opL$. 
    For every $x\in\Lone{\Omega}{\matrices}$, the limit
    \begin{equation}
        \cfrac{1}{N}
        \sum_{n=1}^N
        \seq{\opLdagtau}^n(x)
    \end{equation}
    converges almost surely to some element of $\Lone{\Omega}{\matrices}$. 
    Denoting this limit by $\mfE_{\alpha}^\dagger(x)$, we have that $x\mapsto \mfE_{\alpha}^\dagger(x)$ defines a positive bounded linear operator in $\bops{\Lone{\Omega}{\matrices}}$ with $\banachnorm{\mfE_{\alpha}^\dagger}{\Lone{\Omega}{\matrices}} = 1$ such that $\opLdagtau\circ\mfE_{\alpha}^\dagger = \mfE_{\alpha}^\dagger\circ\opLdagtau  = \mfE_{\alpha}^\dagger = \seq{\mfE_{\alpha}^\dagger}^2$.
\end{prop}
\begin{proof}
For any $a\in\Lone{\Omega}{\matrices}$, by the theorem of Beck and Schwartz \cite{Beck1957ATheorem}, we have that 
\begin{equation}
    \lim_{N\to\infty}
    \frac{1}{N}
    \sum_{n=1}^N
    \seq{\opLdag}^n\!(a)_\omega 1_{\Omega\times\set{0}}\seq{T^n_\alpha(\omega, 0)}
\end{equation}
converges almost everywhere and in $\Lone{\Omega_\alpha}{\matrices}$. 
Thus,
\begin{align}
N_\alpha 
\seq{\lim_{N\to\infty}
    \frac{1}{N}
    \sum_{n=1}^N
    \seq{\opLdag}^n\!(a)_\omega 1_{\Omega\times\set{0}}\seq{T_\alpha^n(\omega, 0)}}
    &\\
    &\hspace{-25mm}=
    \lim_{N\to\infty}
    \seq{
    \cfrac{N^{-1}
    \displaystyle\sum_{n=1}^N
    \seq{\opLdag}^n(a)_\omega 1_{\Omega\times\set{0}}\seq{T_\alpha^n(\omega, 0)}}
    {
    N^{-1}
    \displaystyle\sum_{n=1}^N
    1_{\Omega\times\set{0}}\seq{T_\alpha^n(\omega, 0)}
    }
    }\\
    &\hspace{-25mm}=
    \lim_{N\to\infty}
    \seq{
    \cfrac{
    \displaystyle\sum_{n=1}^N
    \seq{\opLdag}^n(a)_\omega 1_{\Omega\times\set{0}}\seq{T_\alpha^n(\omega, 0)}}
    {
    \displaystyle\sum_{n=1}^N
    1_{\Omega\times\set{0}}\seq{T_\alpha^n(\omega, 0)}
    }
    }
    \\
    &\hspace{-25mm}=
    \lim_{N\to\infty}
     \cfrac{1}{N}
        \sum_{n=1}^N
        \seq{\opLdagtau}^n\!(a)_\omega
\end{align}
exists almost surely and in $\Lone{\Omega}{\matrices}$. 
The rest of the proof follows by standard techniques and we omit it. 
\end{proof}
\begin{lemma}\label{Lem:Inner_prod_mfEalpha}
    Assume $\opL$ is irreducible with unique steady state $\varrho$, and let $\alpha\in\Lambda_\opL$.  
    For any $x\in\Lone{\Omega}{\matrices}$, we have  
    \begin{equation}
        \innerHS{\varrho_\omega}{\mfE_\alpha^\dagger(x)_\omega}
            =
        \avg{\innerHS{\varrho}{x}}
    \end{equation}
    for almost every $\omega\in\Omega$.
\end{lemma}
\begin{proof}
Let $s_n(\omega) = \tau_{\alpha, 1}(\omega) + \cdots + \tau_{\alpha, n}(\omega)$.
We compute 
\begin{align}
     \innerHS{\varrho_\omega}{\mfE_\alpha^\dagger(a)_\omega}
         &=
         \lim_N
         \frac{1}{N}
         \sum_{n=1}^N
         \innerHS{\varrho_\omega}{\seq{\opLdagtau}^n(a)_\omega}\\
         &=
         \lim_N
         \frac{1}{N}
         \sum_{n=1}^N
         \innerHS{\varrho_\omega}{\Phi^{\seq{s_n}\dagger}_\omega(a_{\theta^{s_n}(\omega)})},
\end{align}
so, since $\phi_{\theta(\omega)}(\varrho_\omega) = \varrho_{\theta(\omega)}$ almost surely, we have that 
\begin{equation}
 \innerHS{\varrho_\omega}{\mfE_\alpha^\dagger(a)_\omega}
 =
    \lim_N
         \frac{1}{N}
         \sum_{n=1}^N
         \innerHS{\varrho_{\theta^{s_n}(\omega)}}{a_{\theta^{s_n}(\omega)}}.
\end{equation}
Thus, since 
\begin{align}
    \lim_N
         \frac{1}{N}
         \sum_{n=1}^N
         \innerHS{\varrho_{\theta^{s_n}(\omega)}}{a_{\theta^{s_n}(\omega)}}
    &=
         \lim_N
            \cfrac{\cfrac{1}{N}\displaystyle\sum_{n=1}^N\innerHS{\varrho_{\theta^n(\omega)}}{a_{\theta^n(\omega)}}
            1_{\Omega\times\set{0}}\seq{T_\alpha^n(\omega, 0)}
            }{\cfrac{1}{N}\displaystyle\sum_{n=1}^N
            1_{\Omega\times\set{e}}\seq{T_\alpha^n(\omega, 0)}}\\
    &=
        N_\alpha
        \int_{\Omega_\alpha}
            1_{\Omega\times\set{0}}\innerHS{\varrho}{a}
            \,\dee\mu_\alpha \\
    &= 
    \avg{\innerHS{\varrho}{a}}
\end{align}
    by the ergodicity of $T_\alpha$, the proof is concluded. 
\end{proof}
\begin{cor}\label{Cor:Faithfulness_of_mfE_tau}
    Assume $\opL$ is irreducible and let $\alpha\in\Lambda_\opL$. 
    For any nonzero $x\in\positives{\Lone{\Omega}{\matrices}}$, we have that $\mfE_{\alpha}^\dagger\seq{x}\neq 0$ almost surely, i.e., $\prob{\mfE_{\alpha}^\dagger\seq{x} \neq 0} = 1$.
\end{cor}
\begin{proof}
    By the previous lemma, $\innerHS{\varrho_\omega}{\mfE_\alpha^\dagger(x)_\omega}
        = \avg{\innerHS{\varrho}{x}}$ almost surely. 
    Thus, because $\varrho>0$ almost surely, $x \geq 0$, and $\prob{x \neq 0} > 0$, we conclude that $\avg{\innerHS{\varrho}{x}} > 0$. 
    Thus, $\innerHS{\varrho_\omega}{\mfE_\alpha^\dagger(x)_\omega} > 0$ almost surely, from which the result follows. 
\end{proof}
This concludes our investigation into the case where $\theta$ is fully general. 
Let us take this opportunity to record the following fact of independent interest.
\begin{prop}
    Assume $\opL$ is irreducible. 
    If $\seq{\alpha, \beta, u, f}$ is an eigentuple for $\opL$, then
    \begin{equation}  
        \lim_{N}\frac{1}{N}\sum_{n=1}^N\alpha^{\tau_{\alpha, 1}(\omega) + \cdots + \tau_{\alpha, n}(\omega)}
        =
        \frac{\avg{\expn{N_\alpha}{a_f}}}{\expn{N_\alpha}{a_f(\omega)}}
    \end{equation}
    for a.e. $\omega\in\Omega$.
\end{prop}
\begin{proof}
    Begin by noticing that $\koopman{\theta^{\tau_\alpha}}\!\seq{
        e_{N_\alpha}\!\seq{a_f}
        }
        =
        \alpha^{\tau_\alpha}
        e_{N_\alpha}\!\seq{a_f}$
    almost surely. 
    Thus, 
    \begin{equation}\label{Eqn:Indep_interest_prop}
        \expn{N_\alpha}{a_f(\omega)}\frac{1}{N}\sum_{n=1}^N\alpha^{\tau_{\alpha, 1}(\omega) + \cdots + \tau_{\alpha, n}(\omega)}
            = 
         \frac{1}{N}\sum_{n=1}^N\koopman{\theta^{\tau_\alpha}}^n\!\seq{\expn{N_\alpha}{a_f}}(\omega).
    \end{equation}
    On the other hand, 
    \begin{equation}
         \frac{1}{N}\sum_{n=1}^N\koopman{\theta^{\tau_\alpha}}^n\!\seq{\expn{N_\alpha}{a_f}}(\omega)
        \sim  
        \cfrac{\displaystyle\frac{1}{N}\sum_{n=1}^N
        e_{N_{\alpha}}(a_f(\theta^n(\omega)))1_{\Omega\times\set{0}}(T_\alpha^n(\omega, 0))
        }{\displaystyle\frac{1}{N}\sum_{n=1}^N
        1_{\Omega\times\set{0}}(T_\alpha^n(\omega, 0))
        },
    \end{equation}
    because $1_{\Omega\times\set{0}}(T_\alpha^n(\omega, 0)) = 1$ if and only if $n\in\set{\tau_{\alpha, k}(\omega)}_{k\in\mbN}$, and $\tau_{\alpha, n}<\infty$ almost surely for all $n$. 
    Therefore, because $\lim_N N^{-1}\sum_{n=1}^N
        1_{\Omega\times\set{0}}(T_\alpha^n(\omega, 0)) = N_\alpha^{-1}$ a.e. $\omega$, 
    from (\ref{Eqn:Indep_interest_prop}) and Proposition \ref{Prop:Talpha_is_ergodic}, we conclude 
    \begin{equation}
        \expn{N_\alpha}{a_f(\omega)} \lim_N \frac{1}{N}\sum_{n=1}^N\alpha^{\tau_{\alpha, 1}(\omega) + \cdots + \tau_{\alpha, n}(\omega)}
        =
        N_\alpha 
        \int_{\Omega_\alpha}
            \expn{N_\alpha}{a_f}1_{\Omega\times\set{0}}\,\dee{\mu_\alpha}
        =
        \avg{e_{N_\alpha}\!\seq{a_f}},
    \end{equation}
    concluding the proof. 
\end{proof}
\subsection{Strong irreducibility}
We conclude with a discussion of condition (c) of Theorem \ref{Thmx:Aperiodicity}---the so-called strong irreducibility condition---as it pertains to ergodic quantum processes. 
In light of \cite{Ekblad2024ReducibilityProcesses}, one reasonable way to generalize strong irreducibility to the setting of ergodic quantum processes is the following: 
\begin{definition}[Strong irreducibility]
    Let $\eqp$ be the ergodic quantum process defined by $\seq{\theta, \phi}$, and assume that $\eqp$ is irreducible with unique steady state $\varrho$. 
    We say that $\eqp$ is strongly irreducible for for all $x\in\Lone{\Omega}{\matrices}$, 
    \begin{equation}
        \lim_n \onenorm{
        \avg{\tr{x}}
        \varrho_{\theta^n(\omega)}
            -
        \phi^{(n)}_\omega(x)
        }
        =
        0
    \end{equation}
    holds almost surely. 
\end{definition}
However, the following example shows that this natural notion of strong irreducibility is \textit{not} equivalent to aperiodicity and primitivity, even in the i.i.d. regime. 
\begin{example}[Unitary Haar]\label{Ex:Haar}
Let $\mbU_d$ denote the set of $d\times d$ unitary matrices, and let $\nu$ be the Haar measure on $\mbU_d$. 
For $u\in\mbU_d$, define $\psi_u(a) = uau^*$, which is clearly a quantum channel.
Define $\Omega = \prod_{k\in\mbZ}\mbU_d$, let $\mcF$ be the product $\sigma$-algebra, and let $\mu = \bigotimes_{k\in\mbZ}\nu$ be the i.i.d. measure. 
Let $\theta:\Omega\to\Omega$ be the shift map $\theta\seq{u_k}_{k\in\mbZ} = \seq{u_{k+1}}$. 
Then let $\phi:\Omega\to\channels$ be $\phi_{\seq{u_k}} = \psi_{u_0}$. 
It was shown in \cite{Ekblad2024ReducibilityProcesses} that the ergodic quantum process defined by $\seq{\theta, \phi}$ is irreducible, with unique steady state $\mbI$. 
Because a product of independent Haar unitaries is Haar distributed, Theorem \ref{Thm:PeriodicChar} shows that $|\GammaGroup| = 1$, since we also have that any

However, for any nonrandom projection $p\in\matrices$ with $\tr{p} < d$, we see that 
\begin{equation}
    \phi^{(n)}_{\seq{u_k}}(p)
    =
    u_n\cdots u_1 p u_1^* \cdots u_n^*,
\end{equation}
hence $\phi^{(n)}_{\seq{u_k}}(p)$ is a projection with $\tr{\phi^{(n)}_{\seq{u_k}}(p)} = \tr{p} < d$, so
\begin{equation}
    \onenorm{\tr{p}\mbI - \phi^{(n)}_{\seq{u_k}}(p)} = \tr{\mbI - \phi^{(n)}_{\seq{u_k}}(p)} \geq 1
\end{equation}
almost surely for all $n\in\mbN$. 
\end{example}
These considerations are worth exploring in future work. 
%

%%%

%
\section*{Acknowledgments}
OE is grateful for discussions with Brent Nelson.
Both authors are supported by the National Science Foundation under Grant No. 2153946.
%%%

\section*{Tool and computational resource disclosure}
The proof of Theorem \ref{Thm:Any_group} was developed in consultation with Google Gemini 3.1 Pro. 
\appendix 
%

%%%

%
\section{Disambiguation of irreducibility for unital Schwarz maps of von Neumann algebras}\label{App:Groh irred}
As noted in the preliminaries, $\opL^\dagger_1$ is a unital Schwarz map on the von Neumann algebra $\Linfty{\Omega}{\matrices}$, and in this work we are primarily concerned with understanding the spectral theory of $\opL^\dagger_1$ in the irreducible case. 
In the literature, the spectral theory of some unital Schwarz maps on von Neumann algebras has been previously considered for a certain notion of irreducibility. 
Here, we take a minor detail to disambiguate this diversity of terminology present in the literature. 
\begin{definition}\label{Def:Groh irreducibility}
    We call a positive map $\varphi: \scrA\to \scrA$ on a $C^*$-algebra $\scrA$ \textit{Groh-reducible} (resp. \textit{projection-reducible}) if there is a nontrivial hereditary $C^*$-subalgebra $\scrC\subset \scrA$ (resp. nontrivial projection $p\in\scrA$) such that $\varphi(\scrC)\subset \scrC$ (resp. $\varphi(p\scrA p)\subset p\scrA p$). If no such $\scrC$ (resp. $p$) exists, we call $\varphi$ \textit{Groh-irreducible} (resp. projection-irreducible).
\end{definition}
Thus, what we have called simply \textit{irreducibility} above is for the time being called \textit{projection-irreducibility} (see \cite{Ekblad2024ReducibilityProcesses} for more details).

For a $C^*$-algebra $\scrA$, we say that $\scrF\subset\scrA_+$ is a face if $\scrF$ is a subcone of $\scrA_+$ such that for all $f\in\scrF$ and $a\in\scrA_+$, $a\leq f$ implies $a\in\scrF$.
It is clear that $ \opL^\dagger_1$ being Groh-irreducible is equivalent to there being no nontrivial \textit{norm} closed faces in $\positives{\Linfty{\Omega}{\matrices}}$ invariant under $ \opL^\dagger_1$.
However, from \cite[Proposition 1]{Farenick1996IrreducibleAlgebras}, we see that $ \opL^\dagger_1$ is projection-irreducible if and only if there are no nontrivial \textit{weakly} closed faces in $\positives{\Linfty{\Omega}{\matrices}}$ invariant under $ \opL^\dagger_1$. 
Previously, spectral results for certain unital Schwarz maps on von Neumann algebras were achieved by Groh in \cite{Groh1981TheC-Algebras, Groh1983OnC-Algebras, Groh1984UniformW-Algebras, Groh1984UniformlyC-algebras}. 
Our goal now is to demonstrate the extent to which Groh's result fail to apply in the present situation concerning the operator $\opL^\dagger_1$. 

For an operator $T\in\bops{\scrX}$ on a Banach space $\scrX$, we say $T$ is \textit{quasi-compact} if there is a compact operator $K\in\bops{\scrX}$ and $m\in\mbN$ such that 
    \begin{equation}
        \|
            T^m - K
        \|
            <
        1.
    \end{equation} 
The only fact about quasi-compact operators we use in the following is that the peripheral spectrum of the sort of quasi-compact operators we consider here contains only isolated points, see \cite{Groh1984UniformlyC-algebras}. 
We call $T$ uniformly ergodic if the limit
\begin{equation}
    \lim_{N}\cfrac{1}{N}\sum_{n=1}^NT^n
\end{equation}
exists with respect to the operator norm topology on $\bops{\scrX}$.
\begin{prop}\label{Prop:Groh irreducibility, equivalent characterizations}
Let $\mfT$ be an identity-preserving Schwarz map on a von Neumann algebra $\scrM$. The following are equivalent. 
\begin{enumerate}[label=(\alph*)]
    \item $\mfT$ is Groh-irreducible.

    \item There is a faithful $\mfT^\star$-invariant state $\varphi\in\scrM^\star$ and $\mfT$ is quasi-compact with fixed point space of dimension 1.

    \item There is a faithful $\mfT^\star$-invariant state $\varphi\in\scrM^\star$ and $\mfT$ is uniformly ergodic with fixed point space of dimension 1.
\end{enumerate}
\end{prop}
\begin{proof}
    Let $1$ denote the unit of $\scrM$. 
    The equivalence of (b) and (c) follows from \cite[Theorem 3.3]{Groh1984UniformlyC-algebras}, and that (a) implies (c) is proved in \cite{Groh1984UniformW-Algebras}. So we only need to demonstrate that (c) implies (a). 
    To see this, suppose $\mfT$ is uniformly ergodic and let $\mfP$ denote the uniform limit of the partial sums $\mfP_N = N^{-1}\sum_{n=0}^{N-1}\mfT^n$. Because $\mfT^\star(\varphi) = \varphi$, it holds that $\varphi\circ\mfP_N = \varphi$ for all $N$ hence by continuity of $\varphi$ and the uniform convergence, it holds that $\varphi\circ\mfP = \varphi$. So, if $\scrA\subset\scrM$ is a hereditary $C^*$-subalgebra, we see that $\mfP(\scrA)\neq 0$ if and only if $\scrA\neq 0$ (since $\varphi\circ\mfP(\scrA) = \varphi(\scrA)$ and $\varphi$ is faithful). In particular, if $\scrA$ is a nonzero $\mfT$-invariant $C^*$-subalgebra of $\scrM$, then $\scrA$ is $\mfP$-invariant by the  uniform convergence of $\mfP_N$ to $\mfP$, so because $\scrA$ is the closed span of its projections \cite{Murphy2007C-AlgebrasTheory}, there is some projection $p\neq 0$ in $\scrA$ such that $\mfP(p)\neq 0$. From 
    \begin{equation}
        0 \leq \mfP(p)\leq \mfP(1) = 1
    \end{equation}
    and $\mfT(\mfP(p)) = \mfP(p)$, we see that, if $\scrA\neq\scrM$, then because $\scrA$ is hereditary, it holds that $\mfP(p)\in\scrA$ is not a scalar multiple of $1\in\scrM$. 
    In particular, $\dim\ker(\mfT - 1) \geq 2$, which concludes the proof. 
\end{proof}
As can be seen in the case of $d = 1$ (so that $\matrices = \mbC$) and $(\Omega, \mbP, \theta) = (\mbT, m, \theta)$ where $m$ is the Lebesgue measure and $\theta:\mbT\to\mbT$ is any irrational rotation, the operator $ \opL^\dagger_1 $ is precisely 
\begin{equation}
    \opL^\dagger_1 (f)_\omega = f(\theta(\omega)), 
\end{equation}
the Koopman operator on $L^\infty(\mbT)$.
This operator is clearly an identity-preserving Schwarz map on a von Neumann algebra, and it is also clearly projection-irreducible with 1-dimensional fixed point space. 
However, $ \opL^\dagger_1 $ is not, for example, quasi-compact,
since $\sigma(\koopman{\theta}) = \mbT$. 
That is, $\koopman{\theta}$ is projection-irreducible but not Groh-irreducible. 
In fact, it is not hard to see the following. 
\begin{prop}\label{Prop:Groh irreducibility of L}
    The operator $\opL^\dagger_1$ is Groh-irreducible if and only if $(\Omega, \mu)$ is the union of finitely many atoms of $\mu$ and $\opL^\dagger_1$ is projection-irreducible.
\end{prop}
\begin{proof}
   From the observation that all peripheral spectrum is approximate spectrum and that the diagram 
    \begin{equation}
        \begin{tikzcd}
	{L^\infty(\Omega)\cdot \mbI} && {L^\infty(\Omega)\cdot \mbI} \\
	{L^\infty(\Omega)} && {L^\infty(\Omega)}
	\arrow["{\opL^\dagger_1}", from=1-1, to=1-3]
	\arrow["{\,\,\,\,d^{-1}\operatorname{Tr}}", from=1-1, to=2-1]
	\arrow["{M_\mbI}", from=2-3, to=1-3]
	\arrow["{\koopman{\theta}}", from=2-1, to=2-3]
        \end{tikzcd}
    \end{equation}
    commutes where $M_\mbI$ is the isomorphism defined by $M_\mbI(f) = f\mbI$, we see that 
    \begin{equation}
        \sigma(\opL^\dagger_1)
        \supset 
        \sigma_{\operatorname{ap}}(\opL^\dagger_1\vert_{L^\infty(\Omega)\cdot \mbI})
        =
        \sigma_{\operatorname{ap}}(\koopman{\theta})
        =
        \sigma(\koopman{\theta}).
    \end{equation}
    Thus, if $\opL^\dagger_1$ is Groh-irreducible, it is quasi-compact, hence $\sigma(\koopman{\theta})\neq \mbT$, i.e., $(\Omega, \mbP)$ is the union of finitely many atoms by \cite{Strmer1974SpectraTransformations}. 
    In this setting, we see that Groh-irreducibility is equivalent to projection-irreducibility, because here we have that 
    \begin{equation}
         \Linfty{\Omega}{\matrices} 
        =
        \bigoplus_{k=1}^n
        \mbM_d
    \end{equation}
    for $n$ the number of atoms of $\mu$. 
    In particular, $\Linfty{\Omega}{\matrices} $ is a finite-dimensional von Neumann algebra, and so weakly closed sets are the same as norm closed sets, hence by \cite{Farenick1996IrreducibleAlgebras} Groh-irreducibility is equivalent to projection-irreducibility whenever $(\Omega, \mu)$ is the union of finitely many atoms, which concludes the proof. 
\end{proof}
Thus, the earlier results of Groh do not apply to $\opL^\dagger_1$ except in the extremely restricted situation outlined above, i.e., in the setting that $\theta:\Omega\to\Omega$ is a cyclic permutation of a finite set.

\section{Proof of Proposition \ref{Prop:Cyclic group Gamma when torsion}}\label{App:Group}
Note that $\mbT$ is a \textit{divisible group}, which means that for all $\gamma\in\mbT$ and $n\in\mbN$, there is $\zeta\in\mbT$ such that $\zeta^n = \gamma$. 
This in particular implies that 
\begin{equation}
    \mbT\cong \operatorname{Tor}(\mbT)\oplus\operatorname{Irr}(\mbT),
\end{equation}
where $\operatorname{Tor}(\mbT)$ denotes the set of finite-order elements of $\mbT$ and $\operatorname{Irr}(\mbT) = \mbT/\operatorname{Tor}(\mbT)$.
For any subgroup $\Lambda$ of $\mbT$, we write $\operatorname{Tor}(\Lambda)$ to denote $\operatorname{Tor}(\mbT)\cap\Lambda$ and $\operatorname{Irr}(\Lambda) = \Lambda/\operatorname{Tor}(\Lambda)$. 
We call $\Lambda$ torsion if $\Lambda = \operatorname{Tor}(\Lambda)$. 
Using the fact that $\mbT$ is divisible, it is not hard to see that 
\begin{equation}
    \Lambda \cong \operatorname{Tor}(\Lambda)\oplus \operatorname{Irr}(\Lambda).
\end{equation}
In particular, 
\begin{equation}\label{Eqn:Gamma_direct_sum_decomp}
    \Gamma_\opL 
        \cong 
    \seq{
        \operatorname{Tor}\seq{\Lambda_\opL}/\operatorname{Tor}\seq{\Lambda_\theta}
    }
    \oplus 
    \seq{
        \operatorname{Irr}\seq{\Lambda_\opL}/\operatorname{Irr}\seq{\Lambda_\theta}
    }.
\end{equation}
See \cite[Ch. I, \S 10 \& Ch. XX \S 4]{Lang2002Algebra} for more details. 
With this in hand, we can prove that $\Gamma_\opL$ is cyclic under the assumption that $\Lambda_\theta$ is torsion.
\begin{proof}[Proof of Proposition \ref{Prop:Cyclic group Gamma when torsion}]
     Assume $\Lambda_\theta$ is torsion.
    Since $|\Gamma_\opL|\leq d^2$, 
    from (\ref{Eqn:Gamma_direct_sum_decomp}) we see that $|\operatorname{Irr}(\Lambda_\opL)|<\infty$. 
    But $\operatorname{Irr}(\Lambda_\opL)$ is finite if and only if it is trivial. 
    Thus, $\Lambda_\opL$ is a torsion subgroup of $\mbT$, so there is a nondecreasing sequence of integers $\seq{n_k}_{k\in\mbN}$ with $n_k$ dividing $n_{k+1}$ for all $k$ such that $\Lambda_\opL$ is the direct limit
    \begin{equation}
         \Lambda_\opL \cong \lim_{\substack{\longrightarrow \\ k\in \mbN}}\mbZ/n_k\mbZ.
    \end{equation}
    But $\Lambda_\theta\subset\Lambda_\opL$ is also torsion, so there is another nondecreasing sequence of integers $\seq{m_k}_{k\in\mbN}$ with $m_k$ dividing $n_{k}$ and $m_{k+1}$ for all $k$ with 
    \begin{equation}
        \Lambda_\theta \cong \lim_{\substack{\longrightarrow \\ k\in \mbN}}\mbZ/m_k\mbZ.
    \end{equation}
    So, if we let $\ell_k = n_k/m_k$ for all $k$, we see that 
    \begin{equation}
        \Gamma_\opL
            =
        \Lambda_\opL/\Lambda_\theta 
            =
        \lim_{\substack{\longrightarrow \\ k\in \mbN}}\mbZ/\ell_k\mbZ.
    \end{equation}
    Now, we know $|\Gamma_\opL|\leq d^2$ by Proposition \ref{Prop:Gamma is a finite group}, so this quotient is finite, hence if we let $l_k = \prod_{k'\leq k}\ell_{k'}$, then by 
    \begin{equation}
         \lim_{\substack{\longrightarrow \\ k\in \mbN}}\mbZ/\ell_k\mbZ
         \cong 
          \lim_{\substack{\longrightarrow \\ k\in \mbN}}\mbZ/l_k\mbZ,
    \end{equation}
    the finiteness of $\Gamma_\opL$ implies that there is $K\in\mbN$ such that $l_{K} = l_{K+j}$ for all $j\in\mbN$. 
    In particular, the above gives $\Gamma_\opL\cong\mbZ/l_K\mbZ$, as desired. 
\end{proof}

%

%
%%%
%%%%%
%%%%%%
%%%%%
%%%
%

\printbibliography

\end{document}